\documentclass[11pt]{article}
\usepackage{geometry}
\usepackage{verbatim}
\usepackage{enumitem}
\usepackage{graphicx}
\usepackage{amsmath}
\usepackage{amsthm}
\usepackage{amssymb}
\usepackage{hyperref} 
\usepackage{color}
\usepackage{bbm}
\usepackage{tikz-cd}
\usepackage{subcaption}

\hypersetup{
  colorlinks=true,
  linkcolor=blue,
  filecolor=blue,
  citecolor = black,      
  urlcolor=cyan,
}
\hypersetup{colorlinks=true,citecolor=blue,urlcolor=black,linkbordercolor={1 0 0}}

\makeatletter
\renewcommand{\paragraph}{%
  \@startsection{paragraph}{4}%
  {\z@}{2.25ex \@plus 1ex \@minus .2ex}{-1em}%
  {\normalfont\normalsize\bfseries}%
}
\makeatother

\newcommand\restr[2]{{
  \left.\kern-\nulldelimiterspace 
  #1 
  \vphantom{\big|} 
  \right|_{#2} 
  }}

\newtheorem{theorem}{Theorem}[section]
\newtheorem{lemma}[theorem]{Lemma}

\newtheorem{corollary}[theorem]{Corollary} 

\theoremstyle{definition}
\newtheorem{definition}[theorem]{Definition}
\newtheorem{assumption}[theorem]{Assumption}

\newcommand{\norm}[1]{\left \lVert #1 \right \rVert}

\mathchardef\mhyphen="2D

\DeclareMathOperator*{\argmin}{argmin}

\newcommand{\RR}{\mathbb{R}}

\newcommand{\NN}{\mathbb{N}}

\DeclareMathOperator{\supp}{supp}

\DeclareMathOperator{\sign}{sign}

\newcommand{\EE}{\mathbb{E}}

\DeclareMathOperator{\vspan}{span} 
\DeclareMathOperator{\Proj}{Proj}

\DeclareMathOperator{\polylog}{polylog}

\DeclareMathOperator{\dist}{dist}

\newcommand{\pinv}{\dagger}

\newcommand{\Thetat}{\tilde{\Theta}}

\DeclareMathOperator{\poly}{poly}

\newcommand{\note}[1]{{\color{red} #1}}
\newcommand{\dhruv}[1]{{[\color{purple} dr: #1}]}

\theoremstyle{remark}
\newtheorem{remark}{Remark}
\newcommand{\ignore}[1]{{}}

\title{Distributional Hardness Against Preconditioned Lasso via Erasure-Robust Designs}

\author{Jonathan A. Kelner\thanks{\texttt{kelner@mit.edu}. This work was supported in part by NSF Large CCF-1565235, NSF Medium CCF-1955217, and NSF TRIPODS 1740751.} \\ MIT \and Frederic Koehler\thanks{\texttt{fkoehler@stanford.edu}. This work was supported in part by NSF award CCF-1704417, NSF award IIS-1908774, and N. Anari’s Sloan Research Fellowship} \\ Stanford \and Raghu Meka\thanks{\texttt{raghum@cs.ucla.edu}. This work was supported in part by NSF CAREER Award CCF-1553605 and NSF Small CCF-2007682} \\ UCLA \and Dhruv Rohatgi\thanks{\texttt{drohatgi@mit.edu}. This work was supported by an Akamai Presidential Fellowship and a U.S. DoD NDSEG Fellowship.} \\ MIT}

\begin{document}

\maketitle
\begin{abstract}
    Sparse linear regression with ill-conditioned Gaussian random designs is widely believed to exhibit a statistical/computational gap, but there is surprisingly little formal evidence for this belief, even in the form of examples that are hard for restricted classes of algorithms. Recent work has shown that, for certain covariance matrices, the broad class of Preconditioned Lasso programs provably cannot succeed on polylogarithmically sparse signals with a sublinear number of samples. However, this lower bound only shows that for every preconditioner, there exists at least one signal that it fails to recover successfully. This leaves open the possibility that, for example, trying multiple different preconditioners solves every sparse linear regression problem.
    
    In this work, we prove a stronger lower bound that overcomes this issue. For an appropriate covariance matrix, we construct a single signal distribution on which any invertibly-preconditioned Lasso program fails with high probability, unless it receives a linear number of samples.
    
    Surprisingly, at the heart of our lower bound is a new positive result in compressed sensing. We show that standard sparse random designs are with high probability robust to adversarial measurement erasures, in the sense that if $b$ measurements are erased, then all but $O(b)$ of the coordinates of the signal are still information-theoretically identifiable. To our knowledge, this is the first time that \emph{partial} recoverability of arbitrary sparse signals under erasures has been studied in compressed sensing.
\end{abstract}
\section{Introduction}

The (Gaussian) random-design sparse linear regression (SLR) problem is a fundamental problem in high-dimensional statistics and learning theory. Formally, given independent covariates $X_1,\dots,X_m$ drawn from an $n$-variable Gaussian distribution with known positive-definite covariance $\Sigma$, and responses $y_i = \langle X_i, w^*\rangle$ for some unknown sparse signal $w^*$, the goal is to recover $w^*$. Here, we say that a vector $w^*$ is $k$-sparse if it has at most $k$ nonzero entries. What we have just described is the \emph{realizable}/noiseless version of the problem --- more generally, the response $y_i$ may be assumed to have some amount of stochastic or worst-case noise, but for simplicity we focus on the noiseless case in this paper. As we are proving a lower bound, the noiseless case is in a sense the most interesting and challenging one.

Information-theoretically, for \emph{any} positive-definite covariance $\Sigma$, it's possible to recover $w^*$ exactly from only $O(k \log n)$ samples $(X_i,y_i)$. However, the naive algorithm achieving this sample complexity takes time $O(n^k)$. When $k = \omega(1)$, this is not polynomial time. Moreover, in that regime, all known polynomial-time algorithms for random-design sparse linear regression require either $\Omega(n)$ samples or strong assumptions on the covariance matrix $\Sigma$ \cite{wainwright2019high}. Seemingly, for general $\Sigma$, there is a tradeoff between sample complexity and computational efficiency, commonly referred to as a statistical/computational gap.

Though such a gap is expected to exist in the ill-conditioned setting, there is surprisingly little formal evidence for this gap. While SLR with worst-case covariates (i.e. not drawn from a distribution) has been proven computationally hard under standard worst-case complexity-theoretic assumptions \cite{natarajan1995sparse, zhang2014lower, har2018approximate, gupte2021fine}, no such hardness is known for random-design SLR.\footnote{Here we mean the problem as described with $\Sigma$ invertible, or if $\Sigma$ is not invertible the task would be learning the concept/regression function, i.e. outputting \emph{any} $\hat w$ so that $(\hat w - w)^T\Sigma(\hat w - w) = 0$. If required to output a $k$-sparse predictor and $\Sigma$ is not invertible, the problem is NP hard with infinitely many samples \cite{gupte2020fine} as it's equivalent to finding a sparse solution of linear equations.} Moreover, for average-case problems such as random-design SLR, there are known barriers suggesting that hardness under worst-case assumptions may be unattainable \cite{applebaum2008basing}. Thus, there are broadly two approaches to give formal evidence for a conjectured statistical/computational gap \cite{brennan2020reducibility}: reduction from a conjecturally-hard average-case problem, and unconditional hardness against restricted classes of algorithms.

In recent years, both approaches have yielded evidence for statistical/computational gaps for a variety of \emph{other} statistical problems. Under the Planted Clique conjecture (or a related strengthening of it), problems that are computationally hard include sparse PCA \cite{berthet2013complexity}, average-case RIP certification \cite{wang2016average}, planted dense subgraph \cite{brennan2019universality}, robust sparse mean estimation \cite{brennan2020reducibility}, negative-spike sparse PCA \cite{brennan2020reducibility}, robust isotropic SLR \cite{brennan2020reducibility}, and many others. The literature on hardness of statistical problems against restricted classes of algorithms is even more vast, particularly the branches focused on sum-of-squares algorithms (e.g. \cite{meka2015sum, ma2015sum, kothari2017sum, barak2019nearly}) and statistical query algorithms (e.g. \cite{diakonikolas2017statistical, diakonikolas2019efficient,goel2020superpolynomial, goel2020statistical, dudeja2021statistical}). However, no such evidence has been given for the hardness of non-isotropic random-design SLR; see Section~\ref{section:related-work} for a discussion of prior research in this direction.

One potential explanation for this deficiency is that many of the above approaches either explicitly or implicitly involve constructing distributions over the unknown concept class which encapsulate (to some degree) the hardness of the original problem. For instance, sparse PCA is conjecturally hard even when the planted sparse direction is drawn uniformly at random from $k$-sparse vectors with nonzero entries $\pm 1/\sqrt{k}$ \cite{ma2015sum, brennan2018reducibility}. Statistical query lower bounds\footnote{Relatedly, for realizable regression problems there is a general computationally inefficient algorithm which makes a smaller number of SQ queries \cite{vempala2019gradient}.} are typically derived against the uniform distribution over a finite set of concepts \cite{goel2020superpolynomial}. Thus, it is problematic that for sparse linear regression, no such distribution has even been hypothesized: i.e., a family $(\Sigma_n, \mathcal{D}_n)_{n \in \NN}$ where $\Sigma_n$ describes an $n \times n$ covariance matrix, and $\mathcal{D}_n$ describes a distribution over $k(n)$-sparse $n$-dimensional signals, specifying a model which seems to encapsulate some of the difficulty of random-design SLR. While it would certainly be necessary that the covariance matrices $\Sigma_n$ are ill-conditioned, there are also natural families of ill-conditioned covariate distributions which lead to tractable instances (e.g. \cite{kelner2021power,kelner2019learning}), so such a property is by no means sufficient. Thus, the (open-ended) question is the following: \emph{is there a family $(\Sigma_n,\mathcal{D}_n)_{n \in \NN}$ such that SLR with covariates drawn from $N(0,\Sigma_n)$ and signal drawn from $\mathcal{D}_n$ conjecturally exhibits the computational/statistical gap?}

In this work, we take a step towards answering this question. While we make no conjectures of computational hardness (in particular, it seems quite plausible that there are tailored algorithms for our particular instance), we construct the first distribution family $(\Sigma_n,\mathcal{D}_n)_n$ which provably exhibits non-trivial hardness against a broad class of SLR algorithms. Specifically, we focus on hardness against the recently-introduced class of \emph{Preconditioned Lasso} algorithms, which essentially encompass the current state-of-the-art polynomial-time algorithms for random design SLR. 



\subsection{The Preconditioned Lasso}

The classical approach to solving sparse linear regression is by solving a convex program known as the Lasso \cite{tibshirani1996regression}. In our noiseless setting, it reduces to the basis pursuit program $$\hat{w} \in \argmin_{w \in \RR^n: Xw = y} \norm{w}_1.$$
This program is well-studied, and it's known to succeed with high probability with $O(k\log n)$ samples when $\Sigma$ is well-conditioned \cite{wainwright2019high}. It is also not difficult to construct random-design examples where this program fails with high probability. However, for such examples in the literature \cite{foygelfast, dalalyan2017prediction, kelner2019learning}, it is often possible to apply a sparse change-of-basis after which the covariates are well-conditioned, the signal is still sparse, and therefore the basis pursuit succeeds.

This is one motivation for the definition of a class of algorithms known as Preconditioned Lasso \cite{kelner2021power}, which essentially apply some change-of-basis to ``condition'' the covariates before solving the basis pursuit program. Specifically, for an invertible $n \times n$ matrix $S$ (which we think of as an arbitrary function of $\Sigma$ that however cannot depend on the samples), the $S$-preconditioned Lasso on data $(X_i,y_i)$ applies the transformation $X_i \mapsto S^{-1} X_i$, solves the basis pursuit to get an estimate $\hat{v}$, and returns $\hat{w} := S^T \hat{v}$. This corresponds to solving the convex program \begin{equation} \hat{w} \in \argmin_{w \in \RR^n: Xw = y} \norm{S^T w}_1.\label{eq:S-prec-lasso} \end{equation}

Preconditioned Lasso obviously generalizes the Lasso, and it has been shown to be significantly more powerful; e.g. any covariance matrix with low-treewidth dependency structure induces an SLR model that is tractable via Preconditioned Lasso, even if the matrix is arbitrarily ill-conditioned \cite{kelner2021power}. For this reason, examples that are provably hard against the Preconditioned Lasso are correspondingly more difficult to obtain. Concretely, prior to the present work, the only known hardness result against the Preconditioned Lasso was the following statement:

\begin{theorem}[Informal theorem statement from \cite{kelner2021power}]\label{theorem:prior-lb}
For any $n>0$, there is a positive-definite covariance matrix $\Sigma: n \times n$ such that for any preconditioner $S$, there exists some $\polylog(n)$-sparse signal $w^*$ which $S$-preconditioned Lasso with probability $1-o(1)$ fails to recover, when given $o(n)$ independent samples $X_i \sim N(0,\Sigma)$ and $y_i = \langle X_i,w^*\rangle$.
\end{theorem}

This provides a converse to the algorithmic results of \cite{kelner2021power}, which show that for certain covariance matrices, there is a preconditioner that works for all signals. However, the limitation of Theorem~\ref{theorem:prior-lb} is that the hard signal depends on the preconditioner. 
Thus, it does not provide a hard example for Preconditioned Lasso in the sense described earlier: that is, a covariance matrix and a distribution over signals such that any preconditioner fails with non-trivial probability. 
Indeed, unpacking the proof of Theorem~\ref{theorem:prior-lb} only gives a signal distribution on which any preconditioner fails with probability $\Omega(1/n)$.

\paragraph{Rectangular preconditioners.} To be more precise, Theorem~\ref{theorem:prior-lb} and the definition of Preconditioned Lasso actually apply to all rectangular $n \times s$ preconditioners, not just invertible preconditioners (note that Program~\ref{eq:S-prec-lasso} is still defined, though it no longer corresponds to a change-of-basis). However, most algorithms applying Preconditioned Lasso in the literature use invertible preconditioners \cite{kelner2019learning, kelner2021power}. Moreover, restricting to invertible preconditioners does not improve the failure probability achievable by the techniques in \cite{kelner2021power}. Until the current work, it was hypothetically possible that for any covariance matrix, there are a constant number of changes-of-basis which collectively ``condition'' the covariates, meaning that Lasso succeeds in at least one of the bases.  


\subsection{Main Results}


In this paper, we rule out that possibility. As our first result, we construct a covariance matrix and a sparse signal distribution under which Preconditioned Lasso with any invertible change-of-basis must fail \emph{with high probability}, unless a linear number of samples are given.

\begin{theorem}[Informal statement of Theorem~\ref{theorem:main-invertible}]\label{theorem:invertible-introduction}
Let $n>0$. There is a positive-definite covariance matrix $\Sigma: n \times n$ and a distribution $\mathcal{D}$ over $\polylog(n)$-sparse signals with the following property: for any invertible preconditioner $S$, if we draw $w^* \sim \mathcal{D}$, then $S$-preconditioned Lasso fails to recover $w^*$ with probability at least $1 - o(1)$, when given $o(n)$ independent samples $X_i \sim N(0,\Sigma)$ and $y_i = \langle X_i, w^*\rangle$.
\end{theorem}
In fact, the full version of Theorem~\ref{theorem:main-invertible} is even stronger: it shows that even if we fix a family of 
$\poly(n)$
different preconditioners, they will all fail on a problem instance sampled from our distribution with probability $1 - o(1)$. This allows us to rule out an even larger class of algorithms: for example, the algorithm used in \cite{kelner2019learning} for solving jointly walk-summable SLR instances adaptively selects one out of $n$ possible (invertible) preconditioners before running the Preconditioned Lasso, and our lower bound shows that this strategy and variants are provably defeated by our new construction. 

Additionally, we can extend our result to show hardness against rectangular $n \times s$ preconditioners. For technical reasons we only achieve a failure probability of $1/2 - o(1)$, and require a bound on the preconditioner size. Nonetheless, at this failure probability and with poly-logarithmically sparse signals, we can rule out all polynomially-sized preconditioners.

\begin{theorem}[Informal statement of Theorem~\ref{theorem:main}]\label{theorem:main-introduction}
Let $n>0$. There is a positive-definite covariance matrix $\Sigma: n \times n$ and a distribution $\mathcal{D}$ over $k$-sparse signals, for any $k \geq \log^{12}(n)$, with the following property: for any preconditioner $S$ with at most $\exp(k/\log^{10}(n))$ columns, if we draw $w^* \sim \mathcal{D}$, then $S$-preconditioned Lasso fails to recover $w^*$ with probability at least $1/2 - o(1)$, when given $o(n)$ independent samples $X_i \sim N(0,\Sigma)$ and $y_i = \langle X_i, w^*\rangle$.
\end{theorem}



In both theorems, the probability is over both the signal distribution as well as the random samples.

\subsection{Key Technique: Erasure-Robust Sparse Designs}

In this section we describe the key technique that enables our main results, an erasure-robust sparse design, and provide an independent motivation for this design from compressed sensing. At the end, we provide intuition for how it connects back to hardness against Preconditioned Lasso.

There is a vast literature on sparse linear regression and compressed sensing. Many deterministic conditions and stochastic models for the measurement matrix (also known as design matrix or covariate matrix) have been demonstrated to imply that sparse signals can be recovered either information-theoretically or algorithmically \cite{wainwright2019high}. In noisy settings, the goal is usually either approximate recovery under an $\ell_p$ norm or prediction error (e.g. \cite{candes2006stable, wainwright2009sharp}); exact support recovery with some assumptions about the signal-to-noise ratio (e.g. \cite{wainwright2009sharp}); or approximate support recovery under distributional assumptions about the signal (e.g. \cite{scarlett2016limits}). However, in noiseless settings, the goal is invariably exact recovery. This is obviously ideal. But in situations where the covariates are not entirely under our control, exact recovery could be impossible. A natural goal is then to try to recover \emph{part} of the signal.
To our knowledge, this notion of partial recovery of sparse signals (i.e. due to shortcomings of the measurement matrix rather than due to noise) has received essentially no attention; see Section~\ref{section:related-work} for a discussion of related notions.

Part of the reason may be that it's not obvious what models for a compressive measurement matrix exhibit the behavior that some but not all of the coordinates of a sparse signal are identifiable, besides artificial examples where e.g. unconstrained variables are added to the system. Such examples do not answer the question of whether partial recovery is possible under fundamentally weaker modelling assumptions than total recovery.


\paragraph{Erasure-robustness.} Our key technical contribution is a proof that partial recovery is possible in a natural \emph{semi-random} model. Specifically, we show that random sparse compressive measurement matrices are ``erasure-robust'', by which we mean that if an adversary erases a small fraction of the measurements arbitrarily, then most of the coordinates of the sparse signal vector are still information-theoretically identifiable. Moreover, the identifiability result is stable under inverse-polynomial noise.

Adversarial erasures have been studied in compressing sensing before, and of course have also been long and extensively studied in coding theory (see e.g. \cite{luby2001efficient,langberg2004private,franklin2014optimal}).
For random dense compressive matrices, it's known that deleting a small fraction of the measurements essentially does nothing; the sparse signal is still totally recoverable with high probability \cite{davenport2009simple, voroninski2016strong, lu2019strong}. But for random sparse matrices (which in the absence of erasures do also enable total sparse recovery \cite{berinde2008combining}), no such robustness has been proven, because it's not true: the adversary may simply delete all measurements interacting with a particular coordinate, rendering that coordinate unidentifiable. Given this, partial recovery is the best that can be hoped for. Our result implies that it is also attainable, at least information-theoretically. Moreover, we achieve a nearly-tight bound on the number of unidentifiable coordinates. Here is the informal statement:

\begin{theorem}[Informal statement of Theorem~\ref{theorem:random-is-erasure-robust}]\label{theorem:expander-deletion-intro}
Let $n,m>0$ satisfy $n>m>\Theta(\log^2 n)$, and let $M$ be an $m \times n$ matrix with independent Bernoulli-$p$ entries for $p = \Theta(\log^2 n)/m$. With high probability, the following holds. For any set of ``deleted'' equations $B \subseteq [n/2]$ of size $|B| \leq O(m/\polylog(n))$, there is a set $C$ (the ``unidentifiable coordinates'') of size $|C| \leq 2|B|$ such that $$\norm{x_{C^c}}_2 \leq \poly(n) \cdot \norm{M_{B^c}x}_\infty$$ for any $O(m/\polylog(n))$-sparse vector $x \in \RR^n$.
\end{theorem}

Sparsity of the measurement matrix (and not just the signal) is well-studied in compressed sensing and has various practical applications. For example, in scientific experiments it is often the case that linearity of the response with respect to the covariates is a modelling assumption that's only reasonable for a sparse covariate vector \cite{gilbert2010sparse}. Our result implies that even with a sparse measurement matrix, adversarial erasures (due to e.g. experimental error) are not disastrous. (Note that in the above theorem, each row of the measurement matrix is roughly $(n\log^2 n)/m$ sparse, which up to logarithmic factors cannot be improved, even without erasures, since the measurement matrix must have $\Omega(n)$ nonzero entries).

To clarify the implication of Theorem~\ref{theorem:expander-deletion-intro} for compressed sensing, consider generating $M$ with $\text{Bernoulli}(p)$ entries, and then suppose that some set of measurements with indices $B \subseteq [m]$ is adversarially deleted. We seek to recover a $k$-sparse signal $x^*$ from the matrix $M_{B^c}$ and measurements $$y_{B^c} = M_{B^c}x^* + \eta,$$
where $\norm{\eta}_\infty \leq \delta$. Let $C \subseteq [n]$ be a set with the properties guaranteed by Theorem~\ref{theorem:expander-deletion-intro}; some such set can be found (albeit inefficiently) by brute-force computation of submatrix singular values. Also define the estimator $$\hat{x} \in \argmin_{x: \norm{M_{B^c}x - y_{B^c}}_\infty \leq \delta} \norm{x}_0.$$
We have that $\hat{x} - x^*$ is $2k$-sparse, so if $k \leq O(m/\polylog(n))$, then $\norm{(\hat{x} - x^*)_{C^c}}_2 \leq \delta \cdot \poly(n)$. Thus, we can approximately recover $x^*$ outside the set $C$.

\paragraph{Open question.} Our results show that there is an algorithm for partial sparse recovery in this semi-random model. However, finding a computationally efficient algorithm is an interesting open problem.

\paragraph{Connection to lower bounds against Preconditioned Lasso.} It may seem rather mysterious that construction of an erasure-robust, ``good'' design matrix is the key ingredient in a distributional hard example for a family of sparse recovery algorithms. The technical reasons for this connection are deferred to the overview, but here we try to give some high-level intuition. First, for intuition, we restrict our focus to sparse preconditioners, because dense preconditioners (morally) do not preserve the sparsity of the unknown signal and therefore should not work. Now, if $\Sigma$ is very ill-conditioned, the preconditioner essentially needs to ``fix'' $\Sigma$ by reweighting the different eigenspaces. 

If we are allowed to construct the signal based on the preconditioner, then the preconditioner is forced to be a good approximation for $\Sigma$ everywhere, with no bad directions. But from any good measurement matrix, using the fact that it has dense kernel, we can construct a $\Sigma$ so that no sparse preconditioner can approximate $\Sigma$ everywhere. This is the approach taken in \cite{kelner2021power} to prove Theorem~\ref{theorem:prior-lb}.

In our case, we need to construct the signal distribution without knowing the preconditioner, so the preconditioner is only forced to be a good approximation for $\Sigma$ in \emph{most} directions. Ruling out sparse preconditioners then corresponds to a measurement matrix which has a density property even if some of the rows are ignored (specifically, erasure-robustness). Sparsity of the measurement matrix is needed so that the rows are valid sparse signals, and compressivity is needed so that $\Sigma$ is ill-conditioned in many directions.

\subsection{Related Work}\label{section:related-work}



\paragraph{Hard Examples for SLR.} As we mentioned earlier, SLR with worst-case covariates is known to be computationally hard under worst-case complexity-theoretic assumptions \cite{natarajan1995sparse, zhang2014lower, har2018approximate, gupte2021fine}. However, for the random-design covariate model, there is no known reduction-based hardness, even under average-case or cryptographic assumptions. We now enumerate known restricted hardness results.

First, there is a large literature on when the Lasso and Basis Pursuit programs fail at sparse recovery, even for random designs \cite{wainwright2009sharp, foygelfast, dalalyan2017prediction, van2018tight, kelner2019learning}. While these results do (technically speaking) prove lower bounds against classes of algorithms, these classes are quite small; the constructed examples are only proven to be hard for the Lasso and/or Basis Pursuit (at best, these programs have one meta-parameter). In fact, as has been previously observed \cite{zhang2017optimal, kelner2021power}, all of the ``hard'' examples provided in the above works can be fixed by a simple change-of-basis.

Second, there is a more general lower bound against the class of convex programs solving least-squares regression with coordinate-separable regularizers \cite{zhang2017optimal}. While this is a fairly broad class of algorithms (incomparable with the Preconditioned Lasso), the result has two limitations. One is that the constructed hard signals depend on the regularizer, so there is no single signal distribution that is hard for the entire class. The other limitation is that, like in the previous works on hardness against the Lasso, the hard example in \cite{zhang2017optimal} can be made easy for the Lasso by a simple change-of-basis.

Third, in \cite{kelner2021power}, motivated by the latter limitation, covariance matrices are constructed such that for any change-of-basis, there is a sparse signal (in the original basis) which causes the ``preconditioned'' basis pursuit program to fail. However, as we have previously noted, this result is still limited by the strong dependence of the signal on the preconditioner: it does not even rule out the possibility that there are always \emph{two} preconditioners so that every signal can be recovered by one of them.

Fourth, for \emph{isotropic} random-design SLR (i.e. when $\Sigma = I$), there has been work on identifying the precise sample complexity of sparse recovery. In particular, there appears to be a constant-factor gap between the sample complexities of algorithmic recovery and information-theoretic recovery. Evidence has been given for this gap via the Overlap Gap Property \cite{david2017high}, which implies the failure of a restricted class of ``stable'' algorithms. However, this problem seems fundamentally different from the problem we consider, where the hardness arises from the ill-conditioning of the covariates, and the sample complexity gap is conjecturally exponential rather than a constant.

\paragraph{Partial sparse recovery.} There are several other works in compressed sensing that use the terminology of ``partial'' recovery. To our knowledge, these works all consider different settings from ours; we explain the differences. First, in \cite{bandeira2013partial}, partial sparse recovery refers to totally recovering a signal that is only partially sparse (where the signal space is divided into two sets of coordinates, and it's known that the signal is sparse on the first set). 

Second, in \cite{tajer2012hypothesis}, the goal is indeed to recover only part of the support of the signal. However, their model is the Gaussian Sequence Model (i.e. where the measurement matrix is the identity), where it is obvious that partial recovery is possible, because there is no compression. 

Third, as discussed previously, one common goal in noisy models is partial support recovery (see e.g. \cite{scarlett2016limits}). There, the goal is to estimate the support with few false positives and false negatives, and the reason for error is simply that some coordinates of the signal may be very small and therefore indistinguishable from noise. In contrast, partial identifiability occurs in our setting even without noise, due to a weaker model for the measurement matrix. Moreover, proving partial support recovery in the setting of \cite{scarlett2016limits} requires make strong probabilistic assumptions about the signal, e.g. that the support is a uniform sparse set. In contrast, our results prove conditions under which a measurement matrix enables partial recovery of arbitrary sparse signals.



\subsection{Organization}

In Section~\ref{section:overview}, we provide an overview of the techniques involved in the results of this paper. 
In Section~\ref{section:preliminaries}, we discuss notation and collect important definitions, e.g. of erasure-robustness. In Section~\ref{section:erasure-robustness}, we prove that random sparse compressive matrices are erasure-robust (Theorem~\ref{theorem:expander-deletion-intro}). In Section~\ref{section:structure}, we prove the key structure lemma that connects erasure-robustness with lower bounds against Preconditioned Lasso. In Section~\ref{section:failure}, we then use this lemma together with our result about erasure-robustness to construct example distributions that are provably hard against Preconditioned Lasso algorithms (Theorem~\ref{theorem:invertible-introduction} and Theorem~\ref{theorem:main-introduction}).

\section{Technical Overview}\label{section:overview}

We start with a sketch of the proof of Theorem~\ref{theorem:prior-lb} from \cite{kelner2021power}, which only achieves a failure probability of $O(1/n)$, and which formally motivates the need for erasure-robust sparse designs. We then sketch the proof of our main technical result that random sparse designs are erasure-robust. Finally, we discuss how this result is incorporated into proving stronger lower bounds against Preconditioned Lasso.

\subsection{Lower Bounds via Sparse Designs.}
The hard covariance matrix constructed in \cite{kelner2021power} to prove Theorem~\ref{theorem:prior-lb} is defined as $\tilde{\Sigma} = \tilde{\Theta}^{-1}$ where $\tilde{\Theta} = \Theta + \epsilon I$ and $\Theta = M^T M$, for a rectangular matrix $M$. Note that for small $\epsilon>0$, this covariance is very ill-conditioned, so long as $M$ has non-trivial kernel. However, to actually prove that all Preconditioned Lasso algorithms with $m$ samples fail to recover $k$-sparse signals, these three properties are needed:
\begin{enumerate}
    \item The rows of $M$ are $k$-sparse,
    \item $\dim \ker M \geq 2m$,
    \item $\ker M$ is bounded away from all $(n/k)\log(n)$-sparse vectors.
\end{enumerate}
The first property is self-explanatory. One way to achieve the second property is if $M$ has at most $n-2m$ rows. And the third property, in compressed sensing, is essentially what a design matrix needs to satisfy to information-theoretically enable $(n/k)\log(n)$-sparse recovery. Thus, to show that $\Omega(n)$ samples are needed to recover $\polylog(n)$-sparse signals, $M$ must be a sparse, compressive matrix which (as a design matrix) enables the recovery of $n/\polylog(n)$-sparse signals.

How do these properties imply that for every preconditioner $S$, there is a bad $k$-sparse signal? By the first property, the rows of $M$ are valid signals. For each row $M_i$, if it is not a bad signal for $S$-preconditioned Lasso, then it can be shown to induce a certain constraint on $S$: namely, that every column of $S$ either has small magnitude or is nearly orthogonal to $M_i$. So if none of the rows of $M_i$ are bad signals, then every column of $S$ either has small magnitude or lies near $\ker M$, in which case by the third property it must be $(n/k)\log(n)$-dense. Roughly speaking, this structure can be used together with the second property to show that a $k$-sparse signal with uniformly random support causes the Preconditioned Lasso to fail. 

\paragraph{A hard signal distribution?} The above proof shows that for any preconditioner, either it fails (with high probability) on a random $k$-sparse signal, or there \emph{exists} some row of $M$ on which it fails. If we want a signal distribution that is uniformly hard, it's therefore natural to equiprobably pick either (a) a random row of $M$, or (b) a random $k$-sparse signal. But then the above proof only implies that for this signal distribution, for any preconditioner, the Preconditioned Lasso fails with probability $\Omega(1/n)$. Moreover, it's not clear whether the failure probability can be improved under just the above assumptions: consider the case that for some preconditioner, just a few rows of $M$ are bad signals. Then the likelihood that one of these rows is chosen as the signal is only $O(1/n)$. Moreover, the columns of $S$ are now only forced to be orthogonal to most rows of $M$, not all. As a result, the columns may have large magnitude and yet fail to be dense, because for sparse matrices like $M$, it's possible to adversarially delete a few rows so that the kernel of the remaining rows contains sparse vectors. This is an obstacle to proving that such a preconditioner must fail on a sparse signal with uniformly random support.

To circumvent this obstacle, we need to show that a preconditioner $S$ which has columns orthogonal to most rows of $M$, but not all, still has useful structure. As we suggested earlier, this can be done by reasoning about sparse compressive matrices under adversarial deletions.

\subsection{Erasure-robustness} 
We will return to the lower bound problem in the next section of the overview, but for now focus on the core technical result about partial recovery with adversarial erasures. Based on the discussion after Theorem~\ref{theorem:expander-deletion-intro}, we need to solve the following problem.

Let $M$ be an $m \times n$ sparse random Bernoulli matrix with parameter $p = \Theta(\log n)/m$. We want to show that with high probability, $M$ supports erasure-robust partial sparse recovery: that is, for any set $B \subseteq [m]$ of ``bad equations'', there is a small set $C \subseteq [n]$ such that if $x\in \RR^n$ is $\tau$-sparse, then $$\norm{x_{C^c}}_2 \leq \poly(n) \cdot \norm{M_{B^c}x}_\infty.$$

\paragraph{Erasure-robustness: the exact case.}
For simplicity, in this proof sketch we start by considering the \emph{exact} case, where $M_{B^c}x = 0$, and we want to show that either $|\supp(x)| \geq \tau$ or $\supp(x) \subseteq C$. Without erasures (i.e. $B = \emptyset$), this property follows for $C = \emptyset$ by the fact that the adjacency graph of $M$ is with high probability a unique-neighbor expander.\footnote{We note that this initial part of the argument (the case without erasures) is quite reminiscent of arguments used in the analysis of LDPC codes (see e.g. \cite{sipser1996expander}). }
Concretely, because the graph is a $(1-\epsilon)d$ expander for a small constant $\epsilon>0$, any set $S \subseteq [n]$ of size at most $\tau := O(m/\log(n))$ has at least $(1-O(\epsilon))d|S|$ unique neighbors in $[m]$. Moreover, if $j \in [m]$ is a unique neighbor of $\supp(x)$ for some vector $x \in \RR^n$, then $M_jx \neq 0$. Thus, if $Mx = 0$ then $\supp(x)$ must have no unique neighbors, so either $|\supp(x)| \geq \tau$ or $x=0$.

However, this argument breaks down in the presence of adversarial erasures. All that can be said is that if $M_{B^c}x = 0$ then $\supp(x)$ must have no unique neighbors in $B^c$. By the unique neighbor lower bound, it does follow that either $|\supp(x)| \geq \tau$ or $|\supp(x)| \leq O(|B|/d)$ --- this can be thought of as a kind of \emph{density amplification} result for $\ker M_{B^c}$, since it eliminates the possibility of any vector in the kernel having an intermediate density. Unfortunately, this does not directly imply erasure-robustness, because we need a single set $C$ that contains the supports of all sparse vectors in $\ker M$, not a different $C$ for each $x$. (For example, if we allow $C = \supp(x)$ then the result is not very interesting.) Moreover, it's not clear that anything useful can be said about the vertex set $\supp(x)$: certainly many vertices in $\supp(x)$ must be adjacent to ``bad'' equations, but it's conceivable that other vertices could be farther away. Pictorially, one possible case (of many) is that $B$ could be chosen as the set of ``boundary'' equations of a ball subgraph; then $\ker M_{B^c}$ certainly contains a vector supported on the ball, which is not actually contained in the neighborhood of $B$.

Given the above obstacles, one approach is to show that although $\supp(x)$ may not be contained in the neighborhood of $B$, it must be contained in a distance-$r$ ball around $B$, for some small but super-constant $r$. The argument is that if there is a vertex of $\supp(x)$ which is distance greater than $r$ from $B$, then by iteratively growing neighborhoods of the vertex until $B$ is reached, the support must have size at least $d^r$, and a contradiction is reached if $d^r > |B|/d$, because then $B$ cannot contain all unique neighbors of $\supp(x)$. Unfortunately, the constructed set $C$ (the distance-$r$ ball around $B$) then has size $|B| \cdot (d^2)^r \approx |B|^3$, since the distance metric is that two coordinates are adjacent if they share an equation. This is much larger than the desired bound ($O(|B|)$) and in particular, too large to use in our ultimate lower bound application.

In summary, to get the linear bound claimed in Theorem~\ref{theorem:expander-deletion-intro}, we need a different argument. The key idea is to exploit \emph{linearity}. We want to show that the union $U$ of supports of all $\tau$-sparse vectors in $\ker M_{B^c}$ has small size. We've seen that for any fixed $x$, there is a \emph{density amplification} result: if $M_{B^c}x = 0$ and $x$ is $|B|/d$-dense, then $x$ must be $\tau$-dense. So take vectors $x^{(1)},\dots,x^{(n)} \in \ker M_{B^c}$ which are $\tau$-sparse (and therefore $|B|/d$-sparse) and which cover $U$. Now observe that since $x^{(1)}$ and $x^{(2)}$ are $O(|B|/d)$-sparse, any linear combination $c_1x^{(1)} + c_2 x^{(2)}$ must be $2|B|/d$-sparse. But $c_1x^{(1)} + c_2x^{(2)} \in \ker M_{B^c}$ by linearity. So if $2|B|/d < \tau$, then by the (contrapositive of the) density amplification result, we in fact know that the sum is $|B|/d$ sparse! Inductively, it follows that any linear combination $c_1x^{(1)} + \dots + c_nx^{(n)}$ is $|B|/d$-sparse. But for generic $c_1,\dots,c_n$, we have $$\supp(c_1x^{(1)}+\dots+c_nx^{(n)}) = \bigcup_{i=1}^n \supp(x^{(i)}) = U.$$
This shows that in fact we can find a set $C$ of size $O(|B|/d)$ satisfying the desired property, which is optimal.

\paragraph{Erasure robustness: the general case.} Note that the above argument was when $M_{B^c}x = 0$. The proof for the general case, when $M_{B^c}x$ is small but not nonzero, uses the same insight with several complications. First, we need a quantitative density amplification lemma which states that if $M_{B^c}x$ is small and $x$ has more than $|B|$ coordinates with magnitude exceeding some threshold $\delta$, then we can trade off density for magnitude, i.e. find $\tau$ coordinates with magnitude exceeding $\delta/\poly(n)$. To prove this without losing a superpolynomial factor on the threshold, we actually need the graph to satisfy a stronger property than just expansion: we also need that for any two disjoint sparse sets $S,T \subseteq [n]$, the intersection of their neighborhoods has size only $O(\sqrt{d}\cdot \max(|S|,|T|))$. Note that expansion would only give a bound of $O(\epsilon d \max(|S|,|T|))$. Nonetheless, it can be proven that the random sparse adjacency matrix of $M$ satisfies the desired stronger property with high probability.

Second, the iterative addition procedure in the noiseless case requires a modification for the noisy case; each addition causes the quantitative threshold to decay, and after $n$ additions it would decay by a factor superpolynomial in $n$. Instead, we add the vectors $x^{(1)},\dots,x^{(n)}$ recursively according to a $d$-ary tree. This tree has depth only $\log_d n$, which allows the decay to be controlled to only a $\poly(n)$ factor, proving Theorem~\ref{theorem:expander-deletion-intro}.

\subsection{Stronger lower bound via erasure-robustness}
We now return to the problem of proving hardness against Preconditioned Lasso. Theorem~\ref{theorem:expander-deletion-intro} can be used to show that for an appropriately chosen $M$, if the number of rows of $M$ that are bad signals for $S$-preconditioned Lasso is at most $n/\polylog(n)$, then there is a set $C \subseteq [n]$ of size $n/\polylog(n)$ such that each column of $S$ is either $n/\polylog(n)$-dense, or has small magnitude on coordinates outside the set $C$. This is precisely the structure lemma we need for preconditioners that succeed on most rows of $M$: it crucially allows for a nearly-linear number of rows of $M$ that are bad signals, although in exchange there is a set $C$ of sublinear size where we cannot control the columns of $S$ (the corresponding structure lemma in \cite{kelner2021power} could not tolerate any bad rows). We also show that in this situation, the number of dense columns of $S$ must be $\Omega(n)$. 

With these results, we can prove our lower bounds. First, to prove our lower bound against invertible preconditioners (Theorem~\ref{theorem:invertible-introduction}), we define a distribution over $\polylog(n)$-sparse signals by taking $w^*$ to be the sum of $\polylog(n)$ random rows of $M$, plus an infinitesimal uniformly random $\polylog(n)$-sparse vector. Under certain conditions, if at least one of the rows in the sum is a bad signal, then the sum must also be a bad signal. With this amplification (at the cost of a $\polylog(n)$ factor in sparsity), any invertible preconditioner must fail with probability $1-o(1)$: either there are $\Omega(n/\polylog(n))$ bad rows of $M$, in which case the sum of the chosen rows is a bad signal with high probability, and the infinitesimal perturbation does not affect the program failure. Or, $S$ has many dense columns, in which case $S^T w^*$ is dense due to the perturbation. By a dimension-counting argument (which crucially uses invertibility of $S$), this implies that there exists a feasible direction of improvement for the program objective.

Extending the lower bound to rectangular preconditioners is more involved and involves generalizations of techniques from \cite{kelner2021power}. The factor of $1/2$ in Theorem~\ref{theorem:main-introduction} arises because we are not able to construct a single signal distribution that causes failure of both ``incompatible'' preconditioners (i.e. for which more than $n/\polylog(n)$ rows of $M$ are bad signals) and ``compatible'' preconditioners (for which at most $n/\polylog(n)$ rows of $M$ are bad signals, so the structure lemma applies) with high probability. Instead, we take a mixture of the two cases' hard distributions: either a sum of rows of $M$, or a uniformly random sparse vector. Due to the lack of invertibility of the preconditioners, the second case is no longer a simple dimension-counting argument. In \cite{kelner2021power}, the proof crucially relies on a ``projection lemma'' which states that if $\dim \ker M \geq 2m$, then any fixed direction is unlikely to align with the span of the covariates. Since our structure lemma has no control over the preconditioner columns in the subspace indexed by the set $C$, we prove a generalized projection lemma which states that alignment is unlikely even on $C^c$. This yields Theorem~\ref{theorem:main-introduction}.

\section{Preliminaries}\label{section:preliminaries}

For vectors $x,y$ we denote the inner product as $\langle x,y\rangle = x^T y$. For a matrix $M$, we let $\vspan(M)$ and $\ker(M)$ denote the row space and null space of $M$ respectively. For a matrix $M \in \RR^{n \times p}$, a vector $v \in \RR^n$, and a subset $U \subseteq [n]$, we say that $M_U$ is the $|U| \times p$ matrix consisting of the rows of $M$ indexed by $U$; similarly, $v_U$ consists of the entries of $v$ indexed by $U$; and the complement of $U$ in $[n]$ is $U^c = \bar{U} = [n] \setminus U$. We use the standard notation for vector norms that $\norm{v}_p = \left(\sum_{i=1}^n |v_i|^p\right)^{1/p}$.

\subsection{Preconditioners and the Preconditioned Lasso}\label{section:prelim-preconditioners}

\begin{definition}
For $n,s \in \NN$, a \emph{preconditioner} is a matrix $S \in \RR^{n \times s}$ with $\ker(S^T) = \{0\}$. The $S$-\emph{preconditioned-Lasso} on samples $(X,y)$ is the convex program $$\hat{w} \in \argmin_{w \in \RR^n: Xw = y} \norm{S^T w}_1.$$
\end{definition}

In our lower bounds, we will say that the $S$-preconditioned Lasso \emph{fails} if the true signal vector $w^*$ is not contained in the set of optimal solutions to the program (i.e. some other vector achieves strictly smaller objective value). As a result, some restriction on $\ker(S^T)$ is necessary, to rule out programs with multiple (and therefore infinitely many) co-optimal solutions. In \cite{kelner2021power}, it is only assumed that $S$ is not identically zero; this suffices for their purposes because they are allowed to pick the true signal depending on $S$. However, we want an algorithm-independent distribution over signals. If $S^T$ is nonzero in only a few directions, then it's obviously impossible to cause the $S$-preconditioned Lasso to fail (in our strong sense) with non-trivial probability without knowing $S$. Nonetheless such a program is clearly not useful.

More to the point, any matrix $S$ where $\ker(S^T)$ is non-trivial can be perturbed infinitesimally (possibly by adding columns) so that $\ker(S^T)$ is trivial, but so that if the original program \emph{uniquely} recovered the true signal, then the new program still does so.


\subsection{Supports, erasure-robustness, and quantitative density}

\begin{definition}
For $x \in \RR^n$ and $\delta>0$, define the $\delta$-support of $x$ to be $$\supp_\delta(x) := \{i \in [n]: |x_i| \geq \delta\}.$$
We sometimes refer to $\supp_\delta(x)$ as a \emph{quantitative support} of $x$ with threshold $\delta$. The support of $x$ is $\supp(x) := \norm{x}_0 = \{i \in [n]: |x_i| > 0\}$.
\end{definition}

\begin{definition}\label{def:erasure-robustness}
Let $M \in \RR^{m \times n}$ be a matrix. We say that $M$ is $(b,b',\eta,\tau)$-\emph{erasure-robust} if for any set $B \subseteq [m]$ of size $|B| \leq b$, there is a set $C \subseteq [n]$ of size $|C| \leq b'$ with the following property: for every $x \in \RR^n$, either:
\begin{itemize}
    \item $|\supp(x)| = \norm{x}_0 \geq \tau$, or
    \item $\norm{x_{C^c}}_2 \leq \eta \norm{M_{B^c}x}_\infty$.
\end{itemize}
If $M$ satisfies this property, we say that it tolerates $b$ erasures and sparsity level $\tau$, with only $b'$ unidentifiable coordinates.
\end{definition}

\begin{definition}
For any subspace $V \subseteq \RR^n$ and vector $x \in \RR^n$, the (Euclidean) distance from $x$ to $V$ is $$\dist(x,V) := \inf_{v \in V} \norm{x-v}_2 = \norm{\Proj_{V^\perp} x}_2.$$
\end{definition}

\begin{definition}\label{def:quantitatively-dense}
Let $V \subseteq \RR^n$ be a subspace. We say that $V$ is $(\delta,\eta,\tau)$-\emph{quantitatively dense} if for any set $C \subseteq [n]$ of size $|C| \leq \tau$, for any $x \in \RR^n$ with $\dist(x,V) \leq \delta\norm{x}_2$, it holds that $\norm{x_{C^c}}_2 \geq \eta\norm{x}_2$.
\end{definition}

\subsection{Random Matrix Theory}

We will use the following standard bound on the singular values of Gaussian random matrices.

\begin{theorem}[e.g. Corollary~5.35 in \cite{vershynin2010rmt}]\label{thm:rmt}
    Let $n,N \in \NN$. Let $A \in \RR^{N \times n}$ be a random matrix with entries i.i.d. $N(0,1)$. Then for any $t>0$, it holds with probability at least $1 - 2\exp(-t^2/2)$ that $$\sqrt{N} - \sqrt{n} - t \leq \sigma_\text{min}(A) \leq \sigma_\text{max}(A) \leq \sqrt{N} + \sqrt{n} + t.$$
\end{theorem}

\begin{theorem}[See e.g. Exercise 4.7.3 of \cite{vershynin2018high}]\label{theorem:wishart}
Suppose $X_1,\ldots,X_m \sim N(0,\Sigma)$ with $\Sigma : n \times n$ a positive definite matrix, $t > 0$ and $m = \Omega(n + t^2)$. Let $\hat{\Sigma} = \frac{1}{m} \sum_i X_i X_i^T$. Then with probability at least $1 - 2\exp(-t^2/2)$,
\[ (1 - \epsilon) \Sigma \preceq \hat{\Sigma} \preceq (1 + \epsilon) \Sigma \]
with $\epsilon = O(\sqrt{n/m} + \sqrt{t^2/m})$.
\end{theorem}

\section{Erasure-Robustness via Expanders}\label{section:erasure-robustness}

In this section, we construct a sparse and compressive measurement matrix $M$ which satisfies erasure-robustness with near-linear erasure tolerance and sparsity level, and so that the number of unidentifiable coordinates is only a constant multiple of the number of erasures. In fact, we'll show that a sparse random Bernoulli matrix has the desired property with high probability, proving Theorem~\ref{theorem:expander-deletion-intro} and also enabling our construction of a hard distribution family for Preconditioned Lasso.

Recall that the erasure-robustness property is defined as follows: if $B$ is an adversarially chosen subset of indices of rows of $M$, then $M_{B^c}$ is still a good measurement matrix, meaning that any sparse signal $x$ can still be approximately recovered from noisy measurements $M_{B^c} x + \xi$, except for the entries of $x$ in a small set of input coordinates $C = C(B)$ that is \emph{independent of $x$}. In our construction, we will prove that $|C(B)| =  O(|B|)$, which is nearly optimal because $M$ is sparse, and erasing all equations involving one particular coordinate of the signal clearly makes it impossible to recover that coordinate.


\paragraph{Notation.} 
Let $n,m \in \NN$. Let $M \in \{0,1\}^{m \times n}$ be a binary matrix. The matrix $M$ defines a bipartite graph between a set of ``equations'' $[m]$ and a set of ``coordinates'' $[n]$. For $S \subseteq [n]$ and $E \subseteq [m]$, we can define neighborhoods $$N(S) = \{i \in [m]: \exists j \in S: M_{ij} = 1\}$$ and $$N'(E) = \{j \in [n]: \exists i \in E: M_{ij} = 1\}.$$
For $S \subseteq [n]$, further define the ``unique neighborhood" of $S$ to be $$U(S) = \{i \in [m]: \norm{M_{iS}}_0 = 1\}.$$

\subsection{Deterministic Conditions for Erasure-Robustness}
We start by proving that erasure-robustness holds whenever $M$ satisfies certain deterministic conditions. Throughout this section, we make the following assumption encapsulating the needed properties: approximate regularity, vertex expansion, and a bounded intersection property.

\begin{assumption}\label{assumption:det-expander}
For some $d,k \in \NN$ and $\epsilon > 0$, suppose that $M$ satisfies the following deterministic conditions:
\begin{itemize}
    \item \textbf{(Degree bounds)} For all $i \in [m]$ and $j \in [n]$,
    $$|N(j)| \leq (1+\epsilon)d$$ and $$|N'(i)| \leq (1+\epsilon)(n/m)d.$$
    \item \textbf{(Expansion)} For all $S \subseteq [n]$ with $|S| \leq k$, $$|N(S)| \geq (1-\epsilon)d|S|.$$
    \item \textbf{(Bounded intersection)} For all disjoint $S, T \subseteq [n]$ with $|S|, |T| \leq k$, $$|N(S) \cap N(T)| \leq \frac{\sqrt{d}}{8}\max(|S|,|T|).$$
\end{itemize}
\end{assumption}

Note that the first two conditions imply a weaker form of the bounded intersection property, with $2\epsilon d$ instead of $\sqrt{d}/8$. However, we need the stronger bound to prove our result. 

An immediate corollary of the first two conditions is a unique-neighbor lower bound.

\begin{corollary}[Folklore]\label{corollary:unique-neighbor-lb}
For all $S \subseteq [n]$ with $|S| \leq k$, $$|U(S)| \geq (1-3\epsilon)d|S|.$$
\end{corollary}

\begin{proof}
By the degree bound assumption, we have $\sum_{j \in S} |N(j)| \leq (1+\epsilon)d|S|$. Together with the expansion assumption, it follows that $\sum_{j \in S} |N(j)| - |N(S)| \leq 2\epsilon d|S|$. Every unique neighbor counts once in both terms, and every non-unique neighbor counts at least twice in the first term but only once in the second term. So there are at most $2\epsilon d|S|$ non-unique neighbors, which means that $|U(S)| \geq |N(S)| - 2\epsilon d|S| \geq (1-3\epsilon)d|S|$.
\end{proof}

\paragraph{Density amplification.} For exact solutions $x \in \RR^n$ to $M_{B^c}x = 0$, it's easy to see that the unique neighbor lower bound implies a density amplification statement: if $|\supp(x)| > |B|/((1-3\epsilon)d)$ and $|\supp(x)| \leq k$, then by Corollary~\ref{corollary:unique-neighbor-lb}, $U(\supp(x)) \cap B^c$ must be nonempty, which contradicts $M_{B^c}x = 0$. Thus, $|\supp(x)| > |B|/((1-3\epsilon)d)$ implies $|\supp(x)| > k$.

For approximate solutions, i.e. when $\norm{M_{B^c}x}_\infty \leq \delta$, this argument does not quite work: if some equation $i \in B^c$ has a unique neighbor in $\supp(x)$, this is not necessarily a contradiction because the value of $x$ at that coordinate may be within the error tolerance $\delta$. Thus, we can only hope for amplification if we have a lower bound on $|\supp_\gamma(x)|$, the number of coordinates of $x$ exceeding some threshold $\gamma \gg \delta$ in magnitude. However, there is now another problem: even if some equation has a unique neighbor in $\supp_\gamma(x)$, it may have many neighbors where $x$ has magnitude just slightly less than $\gamma$, so that these coordinates together cancel out the large coordinate. If this happens, then that equation must neighbor both $\supp_\gamma(x)$ and $\supp_{\gamma/(2d)}(x) \setminus \supp_{\gamma}(x)$. Naively, we can use expansion to bound the number of such equations by $2\epsilon d \cdot |\supp_{\gamma/(2d)}(x)|$. However, combined with the unique neighbor lower bound, this only shows that $|\supp_{\gamma/(2d)}(x)| \geq \Omega(1/\epsilon) \cdot |\supp_\gamma(x)| - |B|$. That is, density is amplified by a constant factor, but the threshold decays by a factor of $2d = \Omega(\log n)$.

This tradeoff is not good enough for our purposes (tolerating inverse-polynomial measurement error). The following lemma uses the bounded intersection property to prove a better tradeoff:

\begin{lemma}[One-step quantitative density amplification]\label{lemma:quant-expansion}
Suppose that $\epsilon \leq 1/12$ and $d\geq 2$. Let $x \in \RR^n$ and $\delta>0$. If $\norm{M_{B^c}x}_\infty \leq \delta$, then $$|\supp_\delta x| \geq \min\left(6\sqrt{d}|\supp_{2d\delta} x| - |B|, k\right).$$
\end{lemma}

\begin{proof}
Let $S = \supp_{2d\delta} x$ and let $T = \supp_\delta x \setminus S$. If $|\supp_\delta x| > k$ then the claim holds, so suppose that $|\supp_\delta x| \leq k$. On the one hand, by Assumption~\ref{assumption:det-expander}, we have that $$|N(S) \cap N(T)| \leq \frac{\sqrt{d}}{8}\max(|S|,|T|) \leq \frac{\sqrt{d}}{8}|\supp_\delta x|.$$
On the other hand, we claim that $U(S) \setminus B \subseteq N(S) \cap N(T)$. Clearly $U(S) \setminus B \subseteq N(S)$, so it remains to show inclusion in $N(T)$. Suppose that this is false. Then pick some $i \in U(S) \setminus (B \cup N(T))$. Since $i \in U(S)$, there is a unique $j \in S$ such that $M_{ij} = 1$. Now
\begin{align*}
|M_i x|
&\geq |x_j| - \sum_{j' \neq j} |M_{ij'} x_{j'}| \\
&\geq 2d\delta - \sum_{j' \neq j: M_{ij'} = 1} \delta \\
&\geq (1-\epsilon)d\delta.
\end{align*}
The first inequality is by the triangle inequality. The second inequality uses $j \in S$ and the definition of $S$; it also uses that any $j' \neq j$ with $M_{ij'} = 1$ satisfies $j' \not \in T$ (since $i \not \in N(T)$) and $j' \not \in S$ (since $j$ is the unique neighbor of $i$ in $S$): thus $j' \not \in \supp_\delta x$, so $|x_{j'}| \leq \delta$. The last inequality is by the degree bound on $i$. 

But since $\epsilon < 1-1/d$, the resulting inequality $|M_ix| \geq (1-\epsilon)d\delta$ contradicts the lemma assumption that $\norm{M_{B^c}x}_\infty \leq \delta$, since $i \not \in B$. So in fact the initial claim that $$U(S) \setminus B \subseteq N(S) \cap N(T)$$ was true.
As a result, $$|\supp_\delta x| \geq \frac{8|N(S) \cap N(T)|}{\sqrt{d}} \geq \frac{8|U(S) \setminus B|}{\sqrt{d}} \geq 6\sqrt{d}|S| - |B|$$ where the final inequality uses Corollary~\ref{corollary:unique-neighbor-lb} and the assumption that $\epsilon \leq 1/12$.
\end{proof}

Recursively applying Lemma~\ref{lemma:quant-expansion} gives our quantitative analogue of the noiseless density amplification: if $\norm{M_{B^c}x}_\infty$ is small (i.e. $x$ approximately satisfies the ``good" equations of $M$), then if $x$ has at least $|B|$ not-too-small entries, it must in fact have at least $k$ not-too-small entries. Moreover, the threshold defining ``not-too-small'' only decays by a polynomial factor.

\begin{corollary}[Quantitative density amplification]\label{corollary:density-amplification}
Suppose that $\epsilon \leq 1/12$ and $d \geq 16$. Let $f(n,d) = 2dn^2$. Let $x \in \RR^n$ and $\delta>0$. If $\norm{M_{B^c}x}_\infty \leq \delta$, and $|\supp_{f(n,d)\delta}(x)| \geq |B|$, then $|\supp_\delta(x)| \geq k$.
\end{corollary}

\begin{proof}
For any $\delta' \geq \delta$ with $|\supp_{2d\delta'}(x)| \geq |B|$, Lemma~\ref{lemma:quant-expansion} implies that $$|\supp_{\delta'}(x)| \geq \min\left(\sqrt{2d}|\supp_{2d\delta'}(x)|, k\right).$$
Let $t \in \NN$ be maximal subject to $|\supp_{(2d)^t \delta}(x)| \geq |B|$. Applying the inequality to all $\delta' \in \{\delta,2d\delta,\dots,(2d)^{t-1}\delta\}$ gives that $$|\supp_\delta(x)| \geq \min((\sqrt{2d})^t, k).$$
But since $|\supp_{f(n,d)\delta}(x)| \geq |B|$, we know that $(2d)^{t+1} > f(n,d) = 2dn^2$, so $(\sqrt{2d})^t \geq n$. Thus, $|\supp_\delta(x)| \geq k$.
\end{proof}

\paragraph{Bounding the union of supports.} We've shown that if a single vector $x$ has quantitative support larger than $|B|$ at some threshold, then we can amplify the density by decreasing the threshold (so long as $x$ is an ``approximate solution'', i.e. $\norm{M_{B^c}x}_\infty$ is sufficiently small). Ultimately, we want to bound the union of the quantitative supports over all $k$-sparse vectors for which $M_{B^c}x$ is small. The key step is the following lemma, which shows that adding quantitatively sparse approximate solutions preserves quantitative sparsity (although we have to decrease the threshold by a factor polynomial in the number of terms in the sum).

The lemma is proved recursively: we divide the terms into $O(\sqrt{d})$ groups, apply the inductive hypothesis to bound the supports of the group sums, and then add up the $O(\sqrt{d})$ sums. When there are only $O(\sqrt{d})$ terms, the proof is via Lemma~\ref{lemma:quant-expansion}: if the sum has quantitative support larger than $|B|$, then by amplification it must have quantitative support at least $O(\sqrt{d}|B|)$. But by the triangle inequality, adding up $r$ terms can only increase the quantitative support size by a factor of $r$. For an appropriate choice of constants, this yields a contradiction, proving the desired claim.

\begin{lemma}\label{lemma:sum-density-bound}
Suppose that $\epsilon \leq 1/12$ and $k \geq 4\sqrt{d}|B|$ and $d \geq 16$. Let $\delta>0$ be a threshold. Let $r \leq (\sqrt{2d})^t$ for some $t \in \NN$. Let $x_1,\dots,x_r \in \RR^n$ satisfy $\norm{M_{B^c}x_i}_\infty \leq \delta/\sqrt{2d}$ and $|\supp_\delta x_i| \leq |B|$ for all $i \in [r]$. Then $$\left|\supp_{(2d)^{3t/2}\delta} \sum_{i=1}^r x_i\right| \leq |B|.$$
\end{lemma}

\begin{proof}
We induct on $t$. If $t=0$ then $r \leq 1$, and the claim is trivial. Suppose $t = 1$. For the sake of contradiction, assume that $$\left|\supp_{(2d)^{3/2}\delta} \sum_{i=1}^r x_i\right| > |B|.$$
We know that $$\norm{M_{B^c}\sum_{i=1}^r x_i}_\infty \leq \delta r/\sqrt{2d} \leq \delta.$$ So by Lemma~\ref{lemma:quant-expansion}, $$\left|\supp_{\sqrt{2d}\delta} \sum_{i=1}^r x_i\right| \geq \min(6\sqrt{d}|B| - |B|, k) \geq 4\sqrt{d}|B|.$$
But $$\supp_{\sqrt{2d}\delta}\sum_{i=1}^r x_i \subseteq \bigcup_{i=1}^r \supp_\delta x_i$$ by the triangle inequality and since $r \leq \sqrt{2d}$: if some coordinate $j \in [n]$ is contained in none of the sets $\supp_\delta x_i$, then $|x_{ij}| < \delta$ for all $i \in [r]$, so $\left|\sum_{i=1}^r x_{ij}\right| < r\delta \leq \sqrt{2d}\delta$. Now, the right-hand-side set in the above inclusion has size at most $|B|r \leq \sqrt{2d}|B|$, which contradicts the previous bound we obtained on the size of the left-hand-side set. This proves the lemma statement for $t=1$.

Now suppose $t \geq 2$. Suppose that the lemma holds for all $r \leq (\sqrt{2d})^{t-1}$. Pick some $x_1,\dots,x_r \in \RR^n$ with $r \leq (\sqrt{2d})^t$, satisfying $\norm{M_{B^c}x_i}_\infty \leq \delta/\sqrt{2d}$ and $|\supp_\delta x_i| \leq |B|$ for all $i \in [r]$. Define vectors $y_1,\dots,y_{\sqrt{2d}} \in \RR^n$ by $y_b = \sum_{i=1+(b-1)(\sqrt{2d})^{t-1}}^{\min(b(\sqrt{2d})^{t-1}, r)} x_i$. By the inductive hypothesis applied to each of these smaller sums, we have for all $b \in [\sqrt{2d}]$ that $$|\supp_{(2d)^{3(t-1)/2}\delta} y_b| \leq |B|.$$ Moreover for each $b \in [\sqrt{2d}]$, $$\norm{M_{B^c}y_b}_\infty \leq \sum_{i=1+(b-1)(\sqrt{2d})^{t-1}}^{\min(b(\sqrt{2d})^{t-1},r)} \norm{M_{B^c}x_i}_\infty \leq (\sqrt{2d})^{t-1}\delta/\sqrt{2d} \leq (2d)^{3(t-1)/2}\delta/\sqrt{2d}.$$
By the inductive hypothesis applied to $y_1,\dots,y_{\sqrt{2d}}$ with threshold $(2d)^{3(t-1)/2}\delta$, we have that $$\left|\supp_{(2d)^{3t/2}\delta} \sum_{b=1}^{\sqrt{2d}} y_b \right | \leq |B|$$ as desired.
\end{proof}

The previous lemma lets us bound the quantitative support of a sum of vectors. To show that this bounds the union of the supports, we also need the following lemma, which is proved by the probabilistic method.

\begin{lemma}\label{lemma:random-signs}
Let $\delta>0$. Let $x_1,\dots,x_r \in \RR^n$ and let $V = \bigcup_{i=1}^r \supp_\delta x_i$. Then there is some $\sigma \in \{-1,1\}^r$ such that $$\left|\supp_\delta \sum_{i=1}^n \sigma_i x_i\right| \geq \frac{|V|}{2}.$$
\end{lemma}

\begin{proof}
Suppose we choose $\sigma$ uniformly at random. For any $j \in V$, let $a \in [r]$ be such that $|x_{aj}| \geq \delta$. For any $\sigma$, if $|\sum \sigma_i x_{ij}| < \delta$, then if we flip $\sigma_a$ to obtain a sign vector $\sigma'$, it's necessary that $|\sum \sigma'_i x_{ij}| \geq \delta$. Thus, conditioned on $\sigma_{a}$ it holds that $|\sum_i \sigma_i x_{ij}| \geq \delta$ with probability at least $1/2$, so unconditionally it holds with probability at least $1/2$. Thus, $$\EE_{\sigma \sim \{-1,1\}^r} \left|\supp_\delta \sum_{i=1}^r \sigma_i x_i\right| \geq \frac{|V|}{2}.$$
So there must exist at least one $\sigma \in \{-1,1\}^r$ satisfying the desired inequality.
\end{proof}

We can now bound the union of the quantitative supports of sparse approximate solutions to $M_{B^c}x = 0$, by picking $n$ vectors whose quantitative supports cover the union, and adding them up with the signs suggested by Lemma~\ref{lemma:random-signs}.

\begin{lemma}\label{lemma:union-density-bound}
Suppose that $\epsilon \leq 1/12$ and $|B| \leq k/(4\sqrt{d})$ and $d \geq 16$. Let $\delta>0$. Define $$S = \{x \in \RR^n: \norm{M_{B^c}x}_\infty \leq \delta/\sqrt{2d} \land |\supp_\delta(x)| \leq |B|\}$$ and define $$C = \bigcup_{x \in S} \supp_{g(n,d)\delta}(x)$$
where $g(n,d) = (2nd)^3$. Then $|C| \leq 2|B|$.
\end{lemma}

\begin{proof}
Since $|C| \leq n$, we can find $x_1,\dots,x_n \in \RR^n$ such that $C = \bigcup_{i=1}^n \supp_{g(n,d)\delta}(x_i)$, and $\norm{M_{B^c}x_i}_\infty \leq \delta/\sqrt{2d}$ and $|\supp_\delta(x)| \leq |B|$ for all $i \in [n]$. By Lemma~\ref{lemma:random-signs}, there are some $\sigma_1,\dots,\sigma_n \in \{-1,1\}$ such that $$\left|\supp_{g(n,d)\delta} \sum_{i=1}^n \sigma_i x_i \right| \geq \frac{|C|}{2}.$$ Let $t$ be minimal so that $n \leq (\sqrt{2d})^t$, and note that $(\sqrt{2d})^t \leq 2nd$, so $(2d)^{3t/2} \leq (2nd)^3 = g(n,d)$. By Lemma~\ref{lemma:sum-density-bound}, which is applicable since $\norm{M_{B^c}\sigma_ix_i}_\infty \leq \delta/\sqrt{2d}$ and $|\supp_\delta \sigma_i x_i| \leq |B|$ for all $i \in [n]$, we have that $$\left|\supp_{g(n,d)\delta} \sum_{i=1}^n \sigma_i x_i \right| \leq \left|\supp_{(2d)^{3t/2}\delta} \sum_{i=1}^n \sigma_i x_i \right| \leq |B|.$$
It follows that $|C| \leq 2|B|$.
\end{proof}

From the above lemma, erasure-robustness of $M$ is essentially immediate; it only remains to observe that although we took the union of all approximate solutions with quantitative supports of size at most $|B|$, this is equivalent to taking the union of all approximate solutions with quantitative supports of size less than $k$ (for a slightly different threshold) due to density amplification.

\begin{corollary}[Erasure-robustness under Assumption~\ref{assumption:det-expander}]\label{corollary:dense-or-contained}
Suppose that $\epsilon \leq 1/12$ and $n,d \geq D_0$, for some universal constant $D_0$. For any $b \leq k/d$, the matrix $M$ is $(b,2b, d^4 n^6, k)$-erasure-robust (Definition~\ref{def:erasure-robustness}). To restate: let $B \subseteq [m]$ satisfy $|B| \leq k/d$. Then there is a set $C \subseteq [n]$ with $|C| \leq 2|B|$, such that for any $x \in \RR^n$, either:
\begin{itemize}
\item $\norm{x_{C^c}}_2 \leq d^4 n^6 \norm{M_{B^c}x}_\infty$, or
\item $|\supp(x)| \geq k$.
\end{itemize}
\end{corollary}

\begin{proof}
Pick arbitrary $\delta>0$ and define $S$ and $C$ as in Lemma~\ref{lemma:union-density-bound}, which guarantees that $|C| \leq 2|B|$. For any $x \in \RR^n$, pick some $c \in \RR^+$ such that $\norm{M_{B^c}cx}_\infty \leq \delta/f(n,d)$, where $f(n,d) = 2dn^2$. If $|\supp_\delta(cx)| > |B|$, then by Corollary~\ref{corollary:density-amplification}, $|\supp_{\delta/f(n,d)}(cx)| \geq k$. As a result, $|\supp(x)| = |\supp(cx)| \geq k$. Otherwise, we have $|\supp_\delta(cx)| \leq |B|$. But we also know that $\norm{M_{B^c}cx}_\infty \leq \delta/f(n,d) \leq \delta/d$. Thus, $cx \in S$. As a result, by definition of $C$, we get $\supp_{g(n,d)\delta}(cx) \subseteq C$, where $g(n,d) = (2nd)^3$. Therefore $\norm{cx_{C^c}}_\infty \leq g(n,d)\delta$, so $\norm{cx_{C^c}}_2 \leq \sqrt{n}g(n,d)\delta$.

If $\norm{M_{B^c}x}_\infty > 0$, then we can choose $c = \delta/(f(n,d)\norm{M_{B^c}x}_\infty)$, in which case $\norm{x_{C^c}}_2 \leq \sqrt{n}f(n,d)g(n,d)\norm{M_{B^c}x}_\infty \leq d^4n^6$ for sufficiently large $n,d$. And if $\norm{M_{B^c}x}_\infty = 0$, then we can choose $c$ arbitrarily large, so $\norm{x_{C^c}}_2 = 0$. In either case, it holds that $\norm{x_{C^c}}_2 \leq d^4n^6 \norm{M_{B^c}x}_\infty$.
\end{proof}

All that remains to prove Theorem~\ref{theorem:expander-deletion-intro} is showing that a random sparse binary matrix satisfies Assumption~\ref{assumption:det-expander} with high probability. However, for the application to lower bounds against Preconditioned Lasso, we also must recall the following result (essentially due to \cite{berinde2008sparse} and slightly generalized in \cite{kelner2021power}), which states that under the degree bound and expansion assumptions, any vector in the kernel of $M$ must be quantitatively dense.

\begin{lemma}[Lemma~8.7 in \cite{kelner2021power}]\label{lemma:expander-dense-kernel}
Let $x \in \RR^n$ be such that $Mx = 0$. Let $S \subseteq [n]$ with $|S| \leq k$. If $\epsilon \leq 1/17$, then $$\norm{x_S}_1 \leq \frac{\norm{x}_1}{3}.$$
\end{lemma}

\subsection{Random unbalanced bipartite graph}

In this section we show that a random sparse binary matrix satisfies Assumption~\ref{assumption:det-expander} with high probability. Let $n,d,m \in \NN$ with $d \leq m \leq n$, and let $p = d/m$. Define a random matrix $M \in \RR^{m \times n}$ with independent entries $M_{ij} \sim \text{Ber}(p)$.

The following result is folklore (see e.g. \cite{alon2004probabilistic}):

\begin{lemma}[Expansion of a random bipartite graph]\label{lemma:expansion}
Let $\epsilon \in (0,1)$. Suppose that $p \geq 32\epsilon^{-2}(\log n)/m$. It holds with probability at least $1-2/n$ that for all $S \subseteq [n]$ with $|S| \leq \epsilon/(2p)$, $$|N(S)| \geq d(1-\epsilon)|S|.$$
\end{lemma}

\begin{proof}
For $1 \leq l \leq \epsilon/(2p)$ let $q_l$ be the probability that there exists some $S \subseteq [n]$ with $|S| = l$ and $|N(S)| < d(1-\epsilon)|S|$. To bound this probability, fix $S \subseteq [n]$ with $|S| = l$. For any $y \in [m]$, we have $$\Pr[y \in N(S)] = 1 - (1 - p)^l \geq 1 - e^{-pl} \geq pl - (pl)^2 \geq (1-\epsilon/2)pl$$ so long as $0 \leq pl \leq \epsilon/2$. Thus, by the Chernoff bound, $$\Pr[|N(S)| < (1-\epsilon)plm] \leq \exp(-(\epsilon/2)^2(1-\epsilon/2)plm/2) \leq \exp(-\epsilon^2plm/16).$$
By the union bound over sets of size $l$, if $\epsilon^2pm/16 \geq 2\log n$, we have that $$q_l \leq \binom{n}{l} \exp(-\epsilon^2 plm/16) \leq \exp(l\log n - \epsilon^2 plm/16) \leq \exp(-l\log n).$$
Finally, by a union bound over $1 \leq l \leq \epsilon/(2p)$, the lemma holds with probability at least $$1 - \sum_{l=1}^{\epsilon/(2p)} q_l \geq 1 - \sum_{l=1}^{\epsilon/(2p)} n^{-l} \geq 1 - \frac{2}{n}$$ as claimed.
\end{proof}

We'll also need the following simple result:

\begin{lemma}[Degree bounds]\label{lemma:left-degree}
Let $\epsilon \in (0,1)$. Suppose that $p \geq 6\epsilon^{-2}(\log n)/m$. It holds with probability at least $1 - 1/n$ that $$|N(x)| \leq d(1+\epsilon)$$ for all $x \in [n]$. Similarly, it holds with probability at least $1 - 1/n$ that $|N'(y)| \leq (n/m)d(1+\epsilon)$ for all $y \in [m]$.
\end{lemma}

\begin{proof}
Fix $x \in [n]$. By the Chernoff bound, $$\Pr[|N(x)| > (1+\epsilon)pm] \leq \exp(-\epsilon^2 pm/3).$$
Since $\epsilon^2 pm/3 \geq 2\log n$, this bound is at most $1/n^2$. Union bounding over $x \in [n]$ completes the proof of the first claim.

Similarly, fix $y \in [m]$. By the Chernoff bound, $$\Pr[|\{x: y \in N(x)\}| > (1+\epsilon)pn] \leq \exp(-\epsilon^2 pn/3).$$
Since $\epsilon^2 pn/3 \geq 2\log n$, this bound is at most $1/n^2$, and union bounding over $y \in [m]$ completes the proof.
\end{proof}

To prove the last condition of Assumption~\ref{assumption:det-expander}, we need the following version of the Chernoff-Hoeffding bound:

\begin{lemma}[Chernoff-Hoeffding]\label{lemma:chernoff}
Let $X_1,\dots,X_n$ be i.i.d. random variables with values in $\{0,1\}$. Let $\mu = \EE \sum X_i$. Then for any $t>0$, $$\Pr\left[\sum_{i=1}^n X_i > t\right] \leq \exp(-t \log(t/(\mu e))).$$
\end{lemma}

Now, the proof of the bounded intersection property is analogous to the proof of expansion.

\begin{lemma}
Suppose that $d \geq \log^2 n$ and $n$ is sufficiently large. It holds with probability at least $1 - 2/m^2$ that for all disjoint $S,T \subseteq [n]$ with $|S|,|T| \leq m/d^3$, $$|N(S) \cap N(T)| \leq \frac{\sqrt{d}}{8}\max(|S|,|T|).$$
\end{lemma}

\begin{proof}
Fix $1 \leq l \leq m^3/(n^2d^7)$. Let $q_l$ be the probability that some sets $S,T$ of size exactly $l$ violate the inequality. Fix disjoint $S,T\subseteq [n]$ with $|S|,|T| = l$. For any $i \in [m]$, we have that $$\Pr[i \in N(S)] = 1 - (1-p)^l \leq pl.$$ Thus, since $S$ and $T$ are disjoint, $$\Pr[i \in N(S) \cap N(T)] \leq p^2 l^2.$$
Let $D = \sqrt{d}/8 \geq 3$. By the Chernoff-Hoeffding bound (Lemma~\ref{lemma:chernoff}), we get $$\Pr[|N(S) \cap N(T)| > Dl] \leq \exp(-Dl\log(Dl/(p^2 l^2 m e))) \leq \left(\frac{D}{p^2 l m e}\right)^{-Dl}.$$
By the union bound over sets $S,T$ of size $l$, $$q_l \leq \left(\frac{D}{p^2 l m e}\right)^{-Dl} \left(\frac{en}{l}\right)^{2l} \leq (p^{2D} m^D l^{D-2} n^2)^l = \left(\frac{d^{2D}l^{D-2}n^2}{m^D}\right)^l.$$
So long as $l \leq m/d^3$ and $d \geq \log^2 n$, we have $d^{2D}l^{D-2}n^2/m^D \leq n^2/(m^2 d^{D-6}) \leq 1/m^2$ for large $n$. As a result, summing over $l$, we have $$\sum_{l=1}^{m/d^3} q_l \leq \sum_{l=1}^{m/d^3} (1/m^2)^l \leq 2/m^2$$ as claimed. 
\end{proof}

Together, the above three lemmas immediately imply that the random matrix $M$ satisfies Assumption~\ref{assumption:det-expander} with high probability.

\begin{theorem}\label{theorem:random-is-expander}
Let $n$ be an even positive integer that is sufficiently large. Let $m,d \in \NN$ with $d \leq m \leq n$, and let $\epsilon \in (0,1)$. If $d \geq 32\epsilon^{-2}\log n$ and $d \geq \log^2 n$, then with probability at least $1 - 4/n - 2/m^2$, the random binary matrix $M \in \RR^{m \times n}$ with i.i.d. entries $M_{ij} \sim \text{Ber}(d/m)$ satisfies Assumption~\ref{assumption:det-expander} with sparsity parameter $k = m/d^3$ and error parameter $\epsilon$.
\end{theorem}

Putting together Theorem~\ref{theorem:random-is-expander} and Corollary~\ref{corollary:dense-or-contained} immediately gives the following.

\begin{theorem}\label{theorem:random-is-erasure-robust}
There are constants $N, C$ with the following property. Let $n,m \in \NN$ with $n \geq N$. Let $p \in (0,1)$ satisfy $p \geq C (\log^2 n)/m$. Let $M$ be the $m \times n$ random matrix with independent entries $M_{ij} \sim \text{Ber}(p)$. Then with probability $1 - O(1/m^2)$, it holds that for all $b \leq m/d^4$, the matrix $M$ is $(b, 2b, d^4 n^6, m/d^3)$-erasure-robust.
\end{theorem}

\subsection{Properties of Final Construction}

We can now prove the following existence result, which collects all the important properties of the matrix $M$ which we will use for our lower bound against Preconditioned Lasso: sparsity, quantitative density of the kernel, erasure-robustness, and eigenvalue bounds.

\begin{theorem}\label{theorem:existence}
Let $n \in \NN$ be an even number larger than some absolute constant $n_0$. There is a density parameter $\tau = \Omega(n/\log^6 n)$, $\eta = O(n^6 \log^8 n)$, $b = \Omega(n/\log^8 n)$, and a matrix $M \in \RR^{n/2 \times n}$ with the following properties:
\begin{enumerate}
    \item The rows of $M$ are $O(\log^2 n)$-sparse
    \item For any $x \in \RR^n$ with $\dist(x,\ker M) \leq \norm{x}_2/(12\sqrt{n})$, and any $S \subseteq [n]$ with $|S| \leq \tau$, it holds that $\norm{x_{S^c}}_2 \geq \norm{x}_2/(2\sqrt{n}).$
    \item For any $B \subseteq [n/2]$ with $|B| \leq b$, there is a set $C \subseteq [n]$ with $|C| \leq 2|B|$, such that for any $x \in \RR^n$, either
    \begin{itemize}
        \item $\norm{x_{C^c}}_2 \leq \eta \norm{M_{B^c}x}_2$, or
        \item $|\supp(x)| \geq \tau$.
    \end{itemize}
    \item $\norm{\Theta}_F \leq O(n\log n)$
    \item The smallest nonzero eigenvalue $\lambda$ of $\Theta$ satisfies $\lambda \geq \Omega(n^{-5/2})$.
\end{enumerate}
\end{theorem}

\begin{proof}
Let $\epsilon = 1/17$ and $d = \log^2 n$. Let $M \in \RR^{n/2\times n}$ be the random binary matrix with independent entries $M_{ij} \sim \text{Ber}(2d/n)$. By Theorem~\ref{theorem:random-is-expander}, $M$ satisfies Assumption~\ref{assumption:det-expander} with $k = n/\log^6 n$ and error $\epsilon$, with probability at least $1-12/n$. Let $\tau = k$. Claim (1) follows from the degree bound condition. 

To prove claim (2), let $x \in \RR^n$ with $\dist(x,\ker M) \leq \norm{x}_2/(12\sqrt{n})$, and let $C \subseteq [n]$ with $|C| \leq k$. Let $y = \Proj_{\ker M} x$. Then $My = 0$, so $\norm{y_C}_1 \leq \norm{y}_1/3$ by Lemma~\ref{lemma:expander-dense-kernel}. Therefore
\begin{align*}
\norm{x_{C^c}}_2
&\geq \norm{y_{C^c}}_2 - \frac{1}{12\sqrt{n}}\norm{x}_2 \\
&\geq \frac{1}{\sqrt{n}}\norm{y_{C^c}}_1 - \delta\norm{x}_2 \\
&\geq \frac{2}{3\sqrt{n}}\norm{y}_1 - \frac{1}{12\sqrt{n}}\norm{x}_2 \\
&\geq \frac{2}{3\sqrt{n}}\norm{y}_2 - \frac{1}{12\sqrt{n}}\norm{x}_2 \\
&\geq \frac{1}{2\sqrt{n}}\norm{x}_2
\end{align*}
as desired.

Claim (3) follows from Corollary~\ref{corollary:dense-or-contained} with $b = k/d$. Claim (4) follows from the degree bound assumption, and claim (5) holds with probability at least $1/2$ by results from random matrix theory (Theorem 1.1 in \cite{basak2021sharp}) , and the observation that $\lambda = \sigma^2$, where $\sigma$ is the smallest singular value of $M^T$.

Thus, all claims hold with probability at least $1/2 - 12/n > 0$, so in particular the desired matrix $M$ exists.
\end{proof}

\section{Structure Lemma under Erasure-Robustness}\label{section:structure}

In this section we show how erasure-robust sparse designs $M$ can be used to construct a covariance matrix $\Sigma$ so that preconditioners which are ``compatible'' with most rows of $M$ (i.e. succeed at recovery with non-trivial probability when the covariates are drawn from $N(0,\Sigma)$ and the signal is the row of $M$) satisfy a useful structure lemma. Concretely, the covariance matrix we use to fool the preconditioned Lasso is given by
\[ \tilde \Sigma := \Thetat^{-1}, \qquad \Thetat := \Theta + \epsilon I_n, \qquad \Theta = M^T M \]
where $\epsilon>0$ is polynomially small, so that $\tilde\Sigma$ still has polynomial condition number. Ultimately, we will instantiate $M$ as the matrix constructed in the previous section, but for now we state our results in generality. We will show that if $M$ is erasure-robust and a preconditioner is compatible with most rows of $M$, then the columns of the preconditioner are either dense or have small magnitude outside a small set of coordinates. Moreover, we'll show that the number of dense columns must be linear.

\subsection{Failure on incompatible signals}
To start with, we recall the weak compatibility coefficients $\alpha^{(1)}$ and $\beta^{(1)}$ from \cite{kelner2021power}, which for any fixed covariance matrix and preconditioner, provide a simple necessary condition for the success of the (Preconditioned) Lasso. 
\begin{definition}[Weak $S$-Preconditioned Compatibility Condition \cite{kelner2021power}]\label{def:weak-compatibility} 
We say that $$\alpha^{(1)}_{\Sigma,S,k} = \inf_{w \in B_0(k) \setminus \{0\}} \frac{\langle w, \Sigma w \rangle}{\norm{S^T w}_1^2}$$
where we let $B_0(k)$ denotes the set of $k$-sparse vectors
and
$$\beta^{(1)}_{\Sigma,S,m,k} = \sup\{\beta \in \RR: \dim W_{\Sigma,S,\beta} \geq 2m\}$$
where $$W_{\Sigma,S,\beta} = \left\{w: \langle w, \Sigma w \rangle \geq \beta \norm{S^T w}_1^2\right\},$$
and $\dim W_{\Sigma,S,\beta}$ is defined as the largest dimension of any subspace contained in $W_{\Sigma,S,\beta}$.


\begin{remark}
Since we only deal with invertible covariance matrices $\Sigma$ in this paper, it always holds that $\alpha^{(1)}_{\Sigma,S,k} > 0$. Moreover, so long as $S$ is not identically zero, $\alpha^{(1)}_{\Sigma,S,k}$ is finite. We also assume that $S^T$ has trivial kernel, which implies that $\beta^{(1)}_{\Sigma,S,m,k}$ is finite. This assumption is essentially without loss of generality; see Section~\ref{section:prelim-preconditioners} for discussion. 
\end{remark}

In \cite{kelner2021power}, it shown that if the ratio $\beta^{(1)}_{\Sigma,S,m,k}/\alpha^{(1}_{\Sigma,S,m,k}$ exceeds a certain constant, then there is a $k$-sparse signal such that the $S$-preconditioned Lasso with $m$ samples fails with high probability. However, this signal depends on $S$ (specifically, it's the signal for which $\alpha^{(1)}_{\Sigma,S,k}$ achieves the infimum). Because we ultimately need to construct a signal distribution independent of $S$, we need a slightly more general statement which provides a condition under which a given signal causes $S$-preconditioned Lasso to fail. Such a statement is in fact implicit in \cite{kelner2021power}:

\end{definition}
\begin{theorem}[\cite{kelner2021power}]\label{theorem:ill-conditioned-lasso-failure-l1}
Let $\Sigma \in \RR^{n \times n}$ be positive-definite and let $S \in \RR^{n \times s}$. Let $m,k \in \NN$. If $w^* \in \RR^n$ is a $k$-sparse signal with $$(w^*)^T\Sigma w^* < \frac{\beta^{(1)}_{\Sigma,S,m,k}}{18} \norm{S^T w^*}_1^2,$$ then the $S$-preconditioned Lasso exactly recovers $w^*$ with probability at most $\exp(-\Omega(m))$, from $m$ samples with independent covariates $X_1,\dots,X_m \sim N(0,\Sigma)$ and noiseless responses $Y_i = \langle w^*, X_i \rangle$.
\end{theorem}

\begin{proof}
The proof of this result is implicit in (Theorem~6.5, \cite{kelner2021power}) and we include it for the reader's convenience.
For convenience of notation let $\beta = \beta^{(1)}_{\Sigma,S,m,k}$ and $\Theta = \Sigma^{-1}$. We want to show that with high probability, the $S$-preconditioned Lasso \eqref{eq:S-prec-lasso} fails to recover $w^*$, i.e. $$w^* \not \in \argmin_{w: Xw = Xw^*} \norm{S^T w}_1$$
where $X$ has rows $X_1,\dots,X_m \sim N(0,\Sigma)$. 
We know that $$(w^*)^T \Sigma w^* = \alpha \norm{S^T w^*}_1^2$$ for some $\alpha < \beta/18$.

By definition of $\beta$, there is a subspace $U \subseteq \RR^n$ of dimension $2m$ such that $w^T \Sigma w \geq \beta \norm{S^T w}_1^2$ for all $w \in U$. Let $v_1,\dots,v_{2m} \in U$ form an orthonormal basis for $U$, and let $V \in \RR^{n \times 2m}$ be the matrix with columns $v_1,\dots,v_{2m}$.

We construct $v \in \RR^n$ to satisfy $Xw^* = Xv$ and $\norm{S^T v}_1 < \norm{S^T w^*}_1$ as follows. Let $\Gamma = V^T \Sigma V \in \RR^{2m \times 2m}$. The columns of $V$ have no linear dependencies, and $\Sigma$ is symmetric positive-definite, so $\Gamma$ is symmetric positive-definite. Thus, there is an invertible matrix $N \in \RR^{2m \times 2m}$ such $\Gamma = N^T N$. Define $$c = N^{-1} (XVN^{-1})^\pinv Xw^* \in \RR^{2m}$$ and define $v = Vc \in \RR^n$. By construction we have $v \in U$, so \begin{equation} v^T \Sigma v \geq \beta \norm{S^T v}_1^2.\end{equation}
Second, note that $$\EE[(XVN^{-1})^T(XVN^{-1})] = m(N^{-1})^T V^T \Sigma V N^{-1} = m(N^{-1})^T \Gamma N^{-1} = mI_{2m}.$$
Moreover, the rows of $XVN^{-1}$ are independent and Gaussian. So in fact $XVN^{-1}$ has i.i.d. $N(0,1)$ entries. Thus, with probability $1 - \exp(-\Omega(m))$, we have $\sigma_\text{min}((XVN^{-1})^T) \geq \sqrt{m}/3$ since the dimensions of $(XVN^{-1})^T$ are $2m \times m$ (by Theorem~\ref{thm:rmt}). Hence, $\sigma_\text{max}((XVN^{-1})^\pinv) \leq 3/\sqrt{m}$. We can conclude that \begin{equation} v^T \Sigma v = c^T N^T N c = (w^*)^T X^T (XVN^{-1})^{\pinv T} (XVN^{-1})^\pinv Xw^* \leq (9/m)(w^*)^T X^T X w^*. \end{equation}
We can now check that $\norm{S^Tv}_1 < \norm{S^T w}_1$. Indeed, by Theorem~\ref{theorem:wishart}, $(w^*)^T X^T X w^* \leq 2m (w^*)^T \Sigma w^*$ with probability at least $1 - \exp(-\Omega(m))$, so
\begin{align*}
\norm{S^T v}_1
&\leq \sqrt{\frac{1}{\beta} v^T \Sigma v} \\
&\leq \sqrt{\frac{9}{m\beta} (w^*)^T X^T X w^*} \\
&\leq \sqrt{\frac{18}{\beta} (w^*)^T \Sigma w^*} \\
&\leq \sqrt{\frac{18 \alpha}{\beta}} \norm{S^T w^*}_1
\end{align*}
which produces the desired inequality as long as $\beta/\alpha > 18$.

Finally, since $XVN^{-1}$ is rank-$m$ with probability $1$, we have $(XVN^{-1})(XVN^{-1})^\pinv = I_m$, and thus \begin{equation} Xv = XVN^{-1}(XVN^{-1})^\pinv Xw^* = Xw^* \end{equation}
as desired.
\end{proof}

\subsection{Structure Lemma for compatible preconditioners}
We now consider the case where the preconditioner does a good job of preconditioning the rows $M_i$, so that Lasso can succeed with non-trivial probability for most such signals.
Then very informally, one would think this forces $\Sigma \approx SS^T$, which by the construction of $\Sigma$ will force $S$ to be dense, and hence the preconditioner will fail on basis vectors.

The following Lemma starts to make this intuition correct and precise. It shows that if the preconditioner is compatible with most of the rows of $M$, then the ``sparse part'' of the preconditioner $S$ is not too large. The ``sparse part'' is actually given not just by removing the dense rows of $S^T$ but also by removing a small number $b'$ of columns of $S^T$: this is inevitable because the assumption that the preconditioner is compatible with \emph{most} of the rows of $M$ cannot imply something about \emph{all} of the columns of $S^T$, and this is where erasure-robustness is crucially used. 
\begin{lemma}[Structure Lemma]\label{lem:C-exists}
Let $M \in \RR^{n - r \times n}$, and define $\Theta = M^T M$. Let $\lambda$ be the smallest nonzero eigenvalue of $\Theta$. Let $\epsilon > 0$, and let $\Thetat = \Theta + \epsilon I$. Let $k, m, \alpha,\tau,b,b',\eta > 0$. Suppose that $M$ satisfies $(b,b',\eta,\tau)$-erasure-robustness (Definition~\ref{def:erasure-robustness}).

Let $S \in \RR^{n\times s}$. Suppose that $$\Pr_{i \in [n-r]}\left[M_i^T \Thetat^{-1} M_i < \alpha\norm{S^T M_i}_1^2\right] \leq \frac{b}{n-r}$$
and define $\gamma = \beta^{(1)}_{\Thetat^{-1}, S, k, m}/\alpha$. 
Let
\[ D := \{ i \in [s] : \|(S^T)_i\|_0 \ge \tau \} \]
be the set of $\tau$-dense columns of $S$ (rows of $S^T$). There exists a subset of row indices $C \subseteq [n]$ with $|C| \le b'$ such that the submatrix $S_{C^c D^c}$ satisfies 
$$\sum_{j \in D^c} \norm{S_{C^c j}}_2 \leq \frac{n^{3/2}\eta \norm{M}_F}{\sqrt{\lambda\alpha}}.$$
\end{lemma}

\begin{proof}
Let $B \subseteq [n-r]$ be the set of $i \in [n-r]$ such that $M_i^T \Thetat^{-1} M_i < \alpha \norm{S^T M_i}_1^2$. By assumption, $|B| \leq b$; let $C := C_B \subseteq [n]$ be the set guaranteed by erasure-robustness, which indeed satisfies $|C| \leq b'$.
For any $i \not \in B$, we have 
$$\norm{S^T M_i}_1^2 \leq \frac{1}{\alpha} M_i^T \Thetat^{-1} M_i \leq \frac{\norm{M_i}_2^2}{\lambda\alpha},$$
where the last inequality is because $M_i \in \vspan\Theta$. Expanding the $\ell_1$ norm as a sum, and summing $\norm{S^T M_i}_1$ over $i \not \in B$, we have $$\sum_{j \in [s]} \sum_{i \not \in B} |\langle (S^T)_j, M_i\rangle| \leq \frac{n\sum_{i \not \in B}\norm{M_i}_2}{\sqrt{\lambda \alpha}} \leq \frac{n^{3/2}\norm{M}_F}{\sqrt{\lambda\alpha}}.$$
Thus, $$\sum_{j \in [s]} \norm{M_{B^c}(S^T)_j}_\infty \leq \sum_{j \in [s]} \norm{M_{B^c}(S^T)_j}_1 \leq \frac{n^{3/2}\norm{M}_F}{\sqrt{\lambda\alpha}}.$$
As specified in the theorem statement, let $D \subseteq [s]$ be the set of $i \in [s]$ such that $\norm{(S^T)_i}_0 \geq \tau$. Then for any $i \in D^c$, we have that $\norm{(S^T)_{iC^c}}_2 \leq \eta \norm{M_{B^c}(S^T)_i}_\infty$. Thus, $$\sum_{j \in D^c} \norm{(S^T)_{jC^c}}_2 \leq \eta\sum_{j \in D^c} \norm{M_{B^c}(S^T)_j}_\infty \leq \frac{n^{3/2}\eta\norm{M}_F}{\sqrt{\lambda\alpha}}$$ as claimed.
\end{proof}

As a consequence of the previous lemma, we can show that if the compatibility ratio is small then $S$ must have many dense columns ($S^T$ has many dense rows).

\begin{lemma}\label{lemma:D-size}
In the setting of Lemma~\ref{lem:C-exists}, suppose that $1/\sqrt{\beta^{(1)}_{\Thetat^{-1},S,k,m}\epsilon} > n^{3/2}\eta\norm{M}_F/\sqrt{\lambda\alpha}$. Then $|D| > n-r-b'-2m$.
\end{lemma}

\begin{proof}
Suppose for contradiction that $|D| \leq n-r-b'-2m$. Pick some $\beta' > \beta^{(1)}_{\Thetat^{-1},S,k,m}$ sufficiently small that $1/\sqrt{\beta'\epsilon} > n^{3/2}\eta\norm{M}_F/\sqrt{\lambda\alpha}$. Define $$W = \ker (S^T_D) \cap \ker(\Theta) \cap \vspan\{e_i: i \in C^c\} \subseteq \RR^n.$$
Then $$\dim(W) \geq n - |D| - r - |C| \geq 2m.$$
For any $w \in W$, we have
\begin{align*}
    w^T \Thetat^{-1} w
    &\geq \epsilon^{-1} \norm{\Proj_{\ker\Theta}w}_2^2
    = \epsilon^{-1} \norm{w}_2^2
\end{align*}
since $w \in \ker\Theta$. On the other hand,
\begin{align}
    \norm{S^T w}_1
    &= \sum_{j \in D^c} |\langle (S^T)_j, w\rangle| \tag{$w \in \ker(S^T_D)$} \\
    &= \sum_{j \in D^c} |\langle (S^T)_{jC^c}, w_{C^c}\rangle| \tag{$\supp(w) \subseteq C^c$}\\
    &\leq \sum_{j \in D^c} \norm{(S^T)_{jC^c}}_2\norm{w}_2 \tag{Cauchy-Schwarz} \\
    &\leq \frac{n^{3/2}\eta\norm{M}_F}{\sqrt{\lambda\alpha}} \norm{w}_2 \tag{Lemma~\ref{lem:C-exists}}.
\end{align}
As a consequence, by choice of $\beta'$, we have $w^T\Thetat^{-1}w \geq \beta'\norm{S^T w}_1^2.$ Therefore $W \subseteq W_{\Thetat^{-1},S,\beta'}$, contradicting the definition of $\beta_{\Thetat^{-1},S,k,m}$.
\end{proof}

\section{Failure of the Preconditioned Lasso}\label{section:failure}

In this section, we prove Theorems~\ref{theorem:invertible-introduction} and~\ref{theorem:main-introduction}. We start by proving our lower bound against invertible preconditioners. We then prove a key projection lemma and use it to prove the lower bound against rectangular preconditioners.

\subsection{Invertible Preconditioners}

To construct a signal distribution which fails invertible preconditioners with probability $1-o(1)$, we'll need to amplify the failure probability by adding together multiple signals. The following lemma formalizes why this works in certain cases: a random combination of vectors where at least one of them has large $\ell_1$ norm is very unlikely to have small $\ell_1$ norm.

\begin{lemma}\label{lemma:random-sum-to-max}
Let $\mathcal{D}$ be a continuous distribution on $\RR$ with density upper bounded by $1/2$. Let $v_1,\dots,v_t \in \RR^n$ and let $Z_1,\dots,Z_t \sim \mathcal{D}$ be independent random variables. Then $$\Pr\left[\norm{\sum_{i=1}^t Z_i v_i}_1 < \delta \max_{i \in [t]} \norm{v_i}_1\right] \leq \delta.$$
\end{lemma}

\begin{proof}
Without loss of generality assume that $\norm{v_t}_1 =  \max_{i \in [t]} \norm{v_i}_1$. Condition on $Z_1,\dots,Z_{t-1}$ and define $$f(z) = \norm{\sum_{i=1}^{t-1} Z_i v_i + zv_t}_1.$$
For any $z,z'$ with $f(z) \leq \delta\norm{v_t}_1$ and $f(z') \leq \delta\norm{v_t}_1$ we have by the triangle inequality that $\norm{zv_t - z'v_t}_1 \leq 2\delta\norm{v_t}_1$, so $|z-z'| \leq 2\delta$. Since the density of $\mathcal{D}$ is upper bounded by $1/2$, it follows that $$\Pr[f(Z_t) < \delta\norm{v_t}_1] \leq \frac{2\delta}{2} = \delta$$ which proves the claim.
\end{proof}

We know show that if $M$ is a sparse compressive erasure-robust design matrix (e.g. as we constructed previously), then the covariance matrix $\Sigma = (M^TM + \epsilon I)^{-1}$ together with an appropriate signal distribution describes a hard distribution family for (invertibly) Preconditioned Lasso. 

\begin{theorem}\label{theorem:all-invertible-failure}
Let $n,r \in \NN$. Let $M \in \RR^{n-r\times n}$ and $\epsilon > 0$. Define $\Theta = M^T M$ and $\Thetat = \Theta + \epsilon I$. Let $\lambda$ be the smallest nonzero eigenvalue of $\Theta$. Let $k,m,\alpha,\tau,\eta,b,b',t > 0$ and suppose that $M$ satisfies $(b,b',\eta,\tau)$-erasure-robustness. 

Suppose $k > 2(n/\tau)\log(n)$ and $n-r-b' \geq 3m$. There is a distribution $\mathcal{D}$ on $k(t+1)$-sparse signals in $\RR^n$ with the following property.

Let $S \in \RR^{n \times n}$ be invertible. Suppose that \begin{equation}\label{eqn:cond-req}
9n^{105/2}\eta \norm{M}_F \sqrt{\epsilon/\lambda} < 1.
\end{equation}
Then with probability at least $1 - \exp(-tb/(n-r)) - 1/n^{50}$ over true signals $w^* \sim \mathcal{D}$, it holds that $S$-preconditioned Lasso fails with probability $1-\exp(-\Omega(m))$ over independent samples $X_1,\dots,X_m \sim N(0,\Thetat^{-1})$: that is, $w^*$ is not a unique minimizer of $\norm{S^T w}_1$ subject to $Xw=Xw^*$.
\end{theorem}

\begin{proof}
Let $\mathcal{D}$ be the signal distribution where we draw independent and uniformly random indices $R_1,\dots,R_t \in [n-r]$ as well as independent $Z_0,Z_1,\dots,Z_t \sim \text{Unif}([-1,1])$, and draw $\tilde{w}$ with uniformly random $k$-sparse support and entries $\text{Unif}([-1,1])$ on that support, and set the signal to be $$w^* = Z_0\sqrt{\epsilon}\tilde{w} + \sum_{i=1}^t \frac{Z_iM_{R_i}}{\sqrt{M_{R_i}^T\Thetat^{-1}M_{R_i}}}.$$
Pick any invertible $S \in \RR^{n \times n}$. Define $\alpha = \beta^{(1)}_{\Thetat^{-1},S,k,m}/(72n^{102})$. We distinguish two cases.

\paragraph{Case I: incompatible preconditioner.} On the one hand, suppose that $$\Pr_{i \in [n-r]}\left[M_i\Thetat^{-1}M_i < \alpha \norm{S^T M_i}_1^2\right] \geq \frac{b}{n-r}.$$
Then with probability at least $$\left(1 - (1 - b/(n-r))^t\right) \geq \left(1 - e^{-tb/(n-r)}\right)$$ over the row indices $R_1,\dots,R_t$, there is some $R_i$ with $$\norm{S^T M_{R_i}}_1^2 > \frac{1}{\alpha} M_{R_i}^T\Thetat^{-1} M_{R_i}.$$
Under this event, we have $$\max_{i \in [t]} \frac{\norm{S^T M_{R_i}}_1}{\sqrt{M_{R_i}^T \Thetat^{-1} M_{R_i}}} > \frac{1}{\sqrt{\alpha}},$$
so by Lemma~\ref{lemma:random-sum-to-max}, it holds with probability at least $1-1/n^{50}$ over $Z_1,\dots,Z_t$ that $\norm{S^T w^*}_1 \geq \frac{1}{n^{50}\sqrt{\alpha}}.$ But by the triangle inequality, we have $\sqrt{(w^*)^T \Thetat^{-1} w^*} \leq t+\sqrt{k} \leq 2n$ (since $\tilde{w}^T \Thetat^{-1} \tilde{w} \leq \epsilon^{-1} \norm{\tilde{w}}_2^2$). Thus, $$\Pr_{w^* \sim \mathcal{D}}\left[\norm{S^T w^*}_1^2 > \frac{1}{4n^{102}\alpha}(w^*)^T \Thetat^{-1} w^*\right] \geq \left(1 - e^{-tb/(n-r)}\right)\left(1 - \frac{1}{n^{50}}\right).$$

Moreover, for such $w^*$, by choice of $\alpha$ and by Theorem~\ref{theorem:ill-conditioned-lasso-failure-l1}, the $S$-preconditioned Lasso recovers $w^*$ with probability at most $\exp(-\Omega(m))$, from $m$ independent samples $X_1,\dots,X_m \sim N(0,\Thetat^{-1})$.

\paragraph{Case II: compatible preconditioner.} On the other hand, suppose that $$\Pr_{i \in [n-r]}\left[M_i\Thetat^{-1}M_i \geq \alpha \norm{S^T M_i}_1^2\right] \geq \frac{b}{n-r}.$$ By choice of $\alpha$ and the theorem assumptions, we know that $1/\sqrt{\beta^{(1)}_{\Thetat^{-1},S,k,m}\epsilon} > n^{3/2}\eta\norm{M}_F/\sqrt{\lambda\alpha}$. So we can apply Lemma~\ref{lem:C-exists} and Lemma~\ref{lemma:D-size}: there is a set $D \subseteq [s]$ satisfying $|D| > n-r-b'-2m$, and $\norm{(S^T)_j}_0 \geq \tau$ for all $j \in D$. Let $U = \supp(S^T w^*)$. Then since the support of $\tilde{w}$ is uniformly random of size $k \geq 2(n/\tau)\log(n)$, we know that $U \supseteq \supp(S^T \tilde{w}) \supseteq D$ with probability at least $1 - 1/n$. Let $V = \{d \in \RR^n: \supp(S^T d) \subseteq U\}$ and let $z = \sign(S^T w^*)$. We claim that there is some $d \in V$ with $Xd = 0$ but $\langle d, Sz\rangle \neq 0$. Indeed, since $S^T$ is invertible, it suffices to show that there is a vector $f \in \RR^n$ supported on $U$, such that $\langle S^{-1}X_i, f\rangle = 0$ for all $i \in [m]$, but $\langle z, f\rangle \neq 0$. This holds because $z$ is with probability $1$ outside the span of the vectors $\{(S^{-1}X_i)_U: i \in [m]\}$, and the number of degrees of freedom is $|U| \geq |D| > n-r-b'-2m \geq m$. We conclude that the desired direction of improvement $d$ exists, so $w^*$ is not a minimizer of the $S$-preconditioned Lasso.

\end{proof}
\begin{remark}
The term $1/n^{50}$ in the probability of failure in Theorem~\ref{theorem:all-invertible-failure} and Theorem~\ref{theorem:main-invertible} can be replaced by $1/n^{\ell}$ for any particular $\ell$ if we modify the left hand side of \eqref{eqn:cond-req} accordingly. All that happens if we pick a larger $\ell$ is that to satisfy \eqref{eqn:cond-req}, we need to pick a correspondingly (polynomially) smaller $\epsilon$ which means that the covariance matrix of the data, $\Sigma = \tilde \Theta^{-1}$, becomes more ill-conditioned. We stated the results with $\ell = 50$ only to simplify the statements. 
\end{remark}

Finally, we instantiate the above theorem with the parameters of the design matrix $M$ we constructed in Theorem~\ref{theorem:existence}.

\begin{theorem}\label{theorem:main-invertible}
Let $n \in \NN$ be sufficiently large. There is a matrix $\Sigma \in \RR^{n \times n}$ with condition number $\poly(n)$ such that the following holds. Let $k \in \NN$ with $k \geq \log^8 n$. There is a distribution $\mathcal{D}$ over $O(k\log^9 n)$-sparse signals such that for any positive integer $m \leq n/7$ and any invertible preconditioner $S \in \RR^{n \times n}$, with probability at least $1-O(1/n^{50})$ over $w^* \sim \mathcal{D}$, the $S$-preconditioned Lasso recovers $w^*$ with probability at most $\exp(-\Omega(m))$ from $m$ independent samples $X_1,\dots,X_m \sim N(0,\Sigma)$ with noiseless responses $Y_i = \langle X_i, w^*\rangle$.
\end{theorem}

\begin{proof}
Let $\Theta$ be the matrix guaranteed by Theorem~\ref{theorem:existence}. We check the conditions of Theorem~\ref{theorem:all-invertible-failure}. First, $\dim \ker(\Theta) \geq n/2$. Second, $\lambda = \Omega(n^{-5/2})$. We can take $b = n/(\log^8 n)$, $b' = 2n/(\log^8 n)$, $\eta = n^6 \log^8 n$, and $\tau = \Omega(n/\log^7 n)$. We have $\norm{M}_F \leq O(\sqrt{n\log n})$. Thus, we can take $\epsilon = \Omega(n^{-111})$. We know that the rows of $M$ are $k$-sparse. Let $t = 2\log^9 n$.

Applying Theorem~\ref{theorem:all-invertible-failure}, there is a distribution $\mathcal{D}$ over $O(k\log^9 n)$-sparse signals such that for any invertible preconditioner $S \in \RR^{n \times n}$, with probability at least $1 - O(1/n^{50})$ over $w^* \sim \mathcal{D}$, the $S$-preconditioned Lasso recovers $w^*$ uniquely with probability $\exp(-\Omega(m))$ from $m$ samples, so long as $m \leq n/7$ (so that $n-r-b' \geq 3m$).
\end{proof}

\subsection{Projection Lemma}

To extend our lower bound to rectangular preconditioners, we need Lemma~\ref{lemma:projection-restricted-coordinates}, a projection lemma generalizing an analogous result from \cite{kelner2021power}. To prove it, we recall two lemmas which are essentially taken from \cite{kelner2021power}; the second of these is the original projection lemma.

First, recall that our covariance matrix has the form $\tilde \Sigma = \tilde \Theta^{-1}$, where $\tilde \Theta = \Theta + \epsilon I$ for some PSD matrix $\Theta$. The following lemma establishes that if $\epsilon$ is sufficiently small relative to the smallest nonzero eigenvalue of $\Theta$, then the row span of the design matrix $X$ is nearly orthogonal to all but the top eigenspace of the covariance $\tilde \Sigma = \left(\tilde \Theta\right)^{-1}$ (i.e. the kernel of $\Theta$). In other words, by taking $\epsilon$ small enough the top eigenspace dominates, as expected.
\begin{lemma}[\cite{kelner2021power}]\label{lemma:samples-near-kernel}
Let $\Theta \in \RR^{n \times n}$ be a PSD matrix with minimum nonzero eigenvalue $\lambda$. Let $\epsilon, m > 0$ and let $\Thetat = \Theta + \epsilon I$. Let $X_1,\dots,X_m \sim N(0,\Thetat^{-1})$. If $r := \dim \ker\Theta > 2m$, then with probability at least $1 - \exp(-\Omega(m))$ it holds that for all $a \in \RR^m$, $$\norm{\Proj_{\vspan \Theta} X^T a}_2 \leq C\sqrt{\frac{n\epsilon}{\lambda}} \norm{X^T a}_2$$ where $C$ is an absolute constant, and where $X: m \times n$ is the matrix with rows $X_1,\ldots,X_m$.
\end{lemma}

\begin{proof}
The proof of this result is essentially contained in \cite{kelner2021power}, and we repeat it to make this paper self-contained. 
The statement of the lemma is basis-independent (e.g. does not depend on sparsity of $\Theta$ or $a$), so we can assume without loss of generality that $\Theta$ is diagonal. Then $\Thetat^{-1}$ is diagonal, and we can choose a basis ordering such that the first $r = \dim \ker \Theta$ diagonal entries are each $\epsilon^{-1}$.

Next, note that $(X^T)_{[r]}$ is a $r \times m$ matrix with i.i.d. $N(0, \epsilon^{-1})$ entries, so $\sigma_\text{min}((X^T)_{[r]}) \geq c\epsilon^{-1/2}\sqrt{r}$ with probability at least $1 - \exp(-\Omega(m))$, for some constant $c>0$ (by Theorem~\ref{thm:rmt}). On the other hand, since the entries of $\Thetat^{-1}_{[r]^c,[r]^c}$ are bounded by $1/\lambda$, we also have $\sigma_\text{max}((X^T)_{[r]^c}) \leq C\sqrt{n/\lambda}$ with probability at least $1 - \exp(-\Omega(n))$, for some constant $C$. This means that for any $u \in \RR^m$, $$\norm{(X^T u)_{[r]^c}}_2 \leq C\sqrt{\frac{n}{\lambda}} \norm{u}_2 \leq \frac{C\sqrt{n\epsilon}}{c\sqrt{\lambda r}} \norm{(X^T u)_{[r]}}_2 \leq \frac{C}{c}\sqrt{\frac{n\epsilon}{\lambda}}\norm{X^T u}_2.$$ But $(X^T u)_{[r]^c}$ is precisely the projection of $X^T u$ onto $\vspan \Theta$. So this proves the lemma.
\end{proof}

The next lemma establishes that if the top eigenspace of the covariance $\tilde \Sigma = \left(\tilde \Theta\right)^{-1}$ has a dimension significantly larger than the number of samples, then the row span of the design matrix $X$ is unlikely to align with any particular direction $v$ (i.e. the projection of $v$ onto the null space has large norm). Informally, this is because the worst-case $v$ to consider would be a vector in the top eigenspace (since a lot of the energy of the samples is in this space), and a random lower-dimensional subspace of the top eigenspace (corresponding to the samples) is not likely to contain any particular direction.

\begin{lemma}[Lemma 7.4 in \cite{kelner2021power}]\label{lem:no-alignment}
Let $\Theta \in \RR^{n \times n}$ be a PSD matrix with minimum nonzero eigenvalue $\lambda$. Let $\epsilon, m > 0$ and let $\Thetat = \Theta + \epsilon I$. Let $X_1,\dots,X_m \sim N(0, \Thetat^{-1})$. If $\epsilon \leq c\lambda/n$ for a sufficiently small absolute constant $c>0$, and $r := \dim \ker \Theta > 2m$, then for any fixed $v \in \RR^n$, we have $$\Pr_{X_1,\dots,X_m}[v^T (I - P) v \geq (v^T v)/8] \geq 1 - \frac{4m}{3r} - \exp(-\Omega(m)),$$ where $P = X^T (XX^T)^{-1} X$ is the projection map onto $\vspan\{X_1,\dots,X_m\}$, and where $X: m \times n$ is the matrix with rows $X_1,\ldots,X_m$. As an equivalent statement, it holds with probability at least $1 - (4m)/(3r) - \exp(-\Omega(m))$ that $$\inf_{a \in \RR^m} \norm{v - X^T a}_2 \geq \frac{1}{2\sqrt{2}} \norm{v}_2.$$
\end{lemma}

\begin{proof}
The proof of this result is essentially contained in \cite{kelner2021power}, and we repeat it to make this paper self-contained. 

The statement of the lemma is basis-independent (e.g. does not depend on sparsity of $\Theta$ or $v$), so we can assume without loss of generality that $\Theta$ is diagonal. Then $\Thetat^{-1}$ is diagonal, and we can choose a basis ordering such that the first $r = \dim \ker \Theta$ diagonal entries are each $\epsilon^{-1}$. Let $w =v_{[r]}$ be the first $r$ coordinates of $v$. For $i \in [m]$ let $Y_i = (X_i)_{[r]}$ be the first $r$ coordinates of $X_i$. Then $Y_1,\dots,Y_m$ are i.i.d. $N(0, \epsilon^{-1} I_r)$. So if $P_Y = Y^T (YY^T)^{-1} Y$, then $P_Y$ is projection onto an isotropically random dimension-$m$ subspace of $\RR^r$. Hence, $$\EE \norm{P_Y w}_2^2 = \frac{m}{r} \norm{w}_2^2.$$
With probability at least $1 - 4m/(3r)$ we have $\norm{P_Y w}_2^2 \leq \frac{3}{4}\norm{w}_2^2$. So $\norm{w - P_Y w}_2^2 \geq \norm{w}_2^2/4$. Now $\norm{w-P_Yw}_2$ is the distance from $w$ to the subspace $\vspan\{Y_1,\dots,Y_m\}$. For any vector in $\vspan\{X_1,\dots,X_m\}$, its first $r$ coordinates lie in $\vspan\{Y_1,\dots,Y_m\}$, so the distance to $v$ must be at least $\norm{w-P_Yw}_2$. Thus,
\begin{equation} \norm{v - Pv}_2^2 \geq \norm{w - P_Yw}_2^2 \geq \frac{1}{4} \norm{w}_2^2. \label{eq:dist-to-projection-1}\end{equation}
Next, note that $(X^T)_{[r]}$ is a $r \times m$ matrix with i.i.d. $N(0, \epsilon^{-1})$ entries, so $\sigma_\text{min}((X^T)_{[r]}) \geq c\epsilon^{-1/2}\sqrt{r}$ with probability at least $1 - \exp(-\Omega(m))$, for some constant $c>0$ (by Theorem~\ref{thm:rmt}). On the other hand, since the entries of $\Thetat^{-1}_{[r]^c,[r]^c}$ are bounded by $1/\lambda$, we also have $\sigma_\text{max}((X^T)_{[r]^c}) \leq C\sqrt{n/\lambda}$ with probability at least $1 - \exp(-\Omega(n))$, for some constant $C$. This means that for any $u \in \RR^m$, $$\norm{(X^T u)_{[r]^c}}_2 \leq C\sqrt{\frac{n}{\lambda}} \norm{u}_2 \leq \frac{C\sqrt{n\epsilon}}{c\sqrt{\lambda r}} \norm{(X^T u)_{[r]}}_2.$$ By assumption, $\epsilon \leq (c/4C)^2 \lambda r/n$, so that $\norm{(X^T u)_{[r]^c}}_2 \leq \norm{(X^T u)_{[r]}}_2/4$. Now $Pv$ lies in the span of $X_1,\dots,X_m$, so there is some $u \in \RR^m$ with $Pv = X^T u$. This means that $$\norm{(Pv)_{[r]^c}}_2 \leq \frac{1}{4} \norm{Pv}_2 \leq \frac{1}{4} \norm{v}_2.$$
So $$\norm{v-Pv}_2 \geq \norm{(v - Pv)_{[r]^c}}_2 \geq \norm{v_{[r]^c}}_2 - \frac{1}{4} \norm{v}_2.$$
Together with Equation~\ref{eq:dist-to-projection-1}, which states that $\norm{v-Pv}_2 \geq \frac{1}{2} \norm{v_{[r]}}_2$, we get that $\norm{v-Pv}_2^2 \geq \frac{1}{8} \norm{v}_2^2.$
\end{proof}

Finally, we extend the previous lemma to show that projection of a fixed direction onto the null space of the covariates has lower bounded norm even when restricted to a set of coordinates $P$ with sparse complement. We need an extra condition, that $\ker(\Theta)$ is quantitatively dense (Definition~\ref{def:quantitatively-dense}), meaning that there is no approximate solution $u$ to the equation $u_{P} = 0$ near the top eigenspace of $\tilde \Sigma = \tilde \Theta^{-1}$.
Then we show that under the assumptions of the previous two lemmas, for any particular vector $u_P$ supported in $P$, the equation $u_P = (X^T a)_P$ is unlikely to have an approximate solutions $a$. The proof is by contradiction: if the conclusion is false, then there is a positive probability that if we sample two independent design $X,\bar X$ that: (1) there are approximate solutions $a,\bar a$ to $u_P = (X^T a)_P$ and $u_P = ((\bar X)^T \bar a)_P$, (2) by Lemma~\ref{lemma:samples-near-kernel} both approximate solutions are near the top eigenspace, and (3) by Lemma~\ref{lem:no-alignment} these two solutions are not well-aligned; combining (1-3) gives that $(X^T a - (\bar X)^T \bar a)_P$ is an approximate solution to $(u_P) = 0$ near the top eigenspace, which is impossible.
\begin{lemma}\label{lemma:projection-restricted-coordinates}
There are absolute constants $c,C>0$ such that the following holds. Let $\Theta \in \RR^{n \times n}$ be a PSD matrix with minimum nonzero eigenvalue $\lambda$. Let $\epsilon > 0$ and $\delta,\eta,\tau > 0$ and let $\Thetat = \Theta + \epsilon I$. Let $m \geq C\log n$. Suppose that $\epsilon \leq c\delta^2\lambda/n$, and $r := \dim \ker \Theta > 2m$, and $\ker(\Theta)$ is $(\delta,\eta,\tau)$-quantitatively dense (Definition~\ref{def:quantitatively-dense}). Fix $u \in \RR^n$ and $P \subseteq [n]$ with $|P| \geq n-\tau$. Then with probability at least $p = 1 - (4m)/(3r) - \exp(-\Omega(m))$ over independent $X_1,\dots,X_m \sim N(0,\Thetat^{-1})$, it holds that $$\inf_{a \in \RR^m} \norm{u_P - (X^T a)_P}_2 \geq \frac{\eta}{32}\norm{u_P}_2.$$
\end{lemma}

\begin{proof}
Suppose that the claim is false. Then using that $r = \dim \ker \Theta \le n$ we have
$$\Pr_X\left[\exists a \in \RR^m: \norm{(u - X^T a)_P}_2 < (\eta/32)\norm{u_P}_2\right] > \frac{4m}{3r} \geq \frac{1}{n}.$$
By Lemma~\ref{lemma:samples-near-kernel}, we also have that $$\Pr_X\left[\forall a \in \RR^m: \norm{\Proj_{\vspan \Theta} X^T a}_2 \leq C\sqrt{\frac{n\epsilon}{\lambda}}\norm{X^T a}_2\right] \geq 1 - \exp(-\Omega(m)).$$
Since $m \geq C\log n$, for a sufficiently large constant $C$ the latter probability must exceed $1-1/n$. So with positive probability, these two events occur simultaneously. Hence, there exists some deterministic $\bar{X} \in \RR^{m\times n}$ such that both events occur. Let $\bar{a}$ be the witness of the first event for $\bar{X}$, and define $v = \bar{X}^T \bar{a}$. Then the following two equations hold:
\begin{equation} \norm{(u - v)_P}_2 < \frac{\eta}{32}\norm{u_P}_2 \label{eq:u-to-v} \end{equation}
\begin{equation} \dist(v, \ker \Theta) \leq C\sqrt{\frac{n\epsilon}{\lambda}}\norm{v}_2 \label{eq:v-to-kernel}\end{equation}

Now suppose that the claim of Lemma~\ref{lemma:samples-near-kernel} holds (for the original samples $X_1,\dots,X_m$). Also suppose that the claim of Lemma~\ref{lem:no-alignment} holds for vector $v$ (and samples $X_1,\dots,X_m$), i.e. $v^T(I-P)v \geq (v^T v)/8$: or equivalently, the following inequality holds for arbitrary vectors $a \in \RR^m$:
\begin{equation} 
\norm{v - X^T a}_2 \geq \frac{1}{2\sqrt{2}}\norm{v}_2.\label{eq:v-proj}\end{equation}
By Lemmas~\ref{lemma:samples-near-kernel} and~\ref{lem:no-alignment}, we may assume that both claims hold with probability $p := 1 - (4m)/(3r) - \exp(-\Omega(m))$. Fix $a \in \RR^m$. Then
\begin{align}
\dist(X^T a,\ker \Theta)
&\leq C\sqrt{\frac{n\epsilon}{\lambda}}\norm{X^T a}_2 \tag{by claim of Lemma~\ref{lemma:samples-near-kernel}}\\
&\leq C\sqrt{\frac{n\epsilon}{\lambda}}(\norm{v - X^T a}_2 + \norm{v}_2) \nonumber \\
&\leq 4C\sqrt{\frac{n\epsilon}{\lambda}} \norm{v-X^T a}_2 \tag{by Equation~\ref{eq:v-proj}}
\end{align}
By Equation~\ref{eq:v-to-kernel} and Equation~\ref{eq:v-proj}, $$\dist(v, \ker\Theta) \leq 3C\sqrt{\frac{n\epsilon}{\lambda}} \norm{v-X^T a}_2.$$
Thus, by the triangle inequality, $$\dist(v-X^T a, \ker \Theta) \leq 7C\sqrt{\frac{n\epsilon}{\lambda}} \norm{v-X^T a}_2.$$
By assumption, $7C\sqrt{(n\epsilon)/\lambda} \leq \delta$. So because $\ker(\Theta)$ is $(\delta,\eta,\tau)$-quantitatively dense and $|P| \geq n-\tau$, and by Equation~\ref{eq:v-proj}, $$\norm{(v - X^T a)_P}_2 \geq \eta\norm{v-X^T a}_2 \geq \frac{\eta}{2\sqrt{2}}\norm{v}_2.$$
Finally, we convert this into a bound on $(u - X^T a)_P$, repeatedly using the fact that vectors $u$ and $v$ are close on $P$ (Equation~\ref{eq:u-to-v}):
\begin{align*}
\norm{(u - X^T a)_P}_2
&\geq \norm{(v - X^T a)_P}_2 - \norm{(u-v)_P}_2 \\
&\geq \frac{\eta}{2\sqrt{2}}\norm{v}_2 - \frac{\eta}{32} \norm{u_P}_2 \\
&\geq \frac{\eta}{2\sqrt{2}}\norm{v_P}_2 - \frac{\eta}{32} \norm{u_P}_2 \\
&\geq \frac{\eta}{2\sqrt{2}} \norm{u_P}_2 - \frac{\eta}{2\sqrt{2}} \norm{(u-v)_P}_2 - \frac{\eta}{32}\norm{u_P}_2 \\
&\geq \frac{\eta}{8}\norm{u_P}_2
\end{align*}
where the second and last inequalities apply Equation~\ref{eq:u-to-v}. We showed this inequality holds for all $a \in \RR^m$ with probability at least $p$, which is the desired conclusion of the Lemma. This contradicts the initial assumption that the conclusion is false, proving the conclusion unconditionally.
\end{proof}

\subsection{Failure of rectangular preconditioners}

\begin{lemma}\label{lemma:large-row-fraction}
Let $r := \dim \ker \Theta$. Then for any $S \in \RR^{n \times s}$, $$\Pr_{i \in [n]}\left[\norm{S^T e_i}_1 \geq \frac{1}{\sqrt{\beta\epsilon n}}\right] > \frac{r-2m}{n}$$ where $\beta = \beta_{\Thetat^{-1},S,k,m}$.
\end{lemma}

\begin{proof}
Suppose not. Then $\norm{S^T e_i} < 1/\sqrt{\beta\epsilon n}$ for at least $n+2m-r$ choices of $i \in [n]$. Pick some $\beta' > \beta$ which is sufficiently close to $\beta$ that $I = \{i \in [n]: \norm{S^T e_i} < 1/\sqrt{\beta'\epsilon n}\}$ also satisfies $|I| \geq n+2m-r$. Let $V = \vspan\{e_i: i \in I\}$. Then $\dim V = |I| \geq n + 2m - r$. Define $W = V \cap \ker \Theta$. Then $\dim W \geq \dim V - (n-r) \geq 2m$. Moreover, for any $w \in W$, we have
\begin{align*}
w^T \Thetat^{-1} w
&\geq \epsilon^{-1} \norm{\Proj_{\ker\Theta} w}_2^2 \\
&= \epsilon^{-1} \norm{w}_2^2
\end{align*}
whereas
$$\norm{S^T w}_1 \leq \sum_{i \in I} |w_i|\cdot \norm{S^T e_i}_1 \leq \frac{1}{\sqrt{\beta'\epsilon n}} \norm{w}_1 \leq \frac{\norm{w}_2}{\sqrt{\beta'\epsilon}}.$$
As a consequence, $w^T\Thetat^{-1}w \geq \beta'\norm{S^T w}_1^2.$ Therefore $W \subseteq W_{\Thetat^{-1},S,\beta'}$, contradicting the definition of $\beta_{\Thetat^{-1},S,k,m}$.
\end{proof}

\begin{lemma}\label{lemma:compatible-preconditioner-failure}
There are constants $c,c_m>0$ so that the following holds. Let $M \in \RR^{n-r \times n}$ and $\epsilon > 0$. Define $\Theta = M^T M$ and $\Thetat = \Theta + \epsilon I$. Let $\lambda$ be the smallest nonzero eigenvalue of $\Theta$. Let $k,m,s,\alpha,\tau,\eta,b,b' > 0$ and suppose that $M$ satisfies $(b,b',\eta,\tau)$-erasure-robustness, and $\ker(M)$ is $(1/(12\sqrt{n}), 1/2\sqrt{n}, b')$-quantitatively-dense. 
Suppose $r > 2m$, $m \geq c_m\log n$, $k > (n/\tau)\log(sn)$ and $\epsilon < c\lambda/n^2$. Let $\mathcal{D}_k$ be the distribution of $k$-sparse signals with uniformly random support in $\RR^n$, and Gaussian entries on the support. For any preconditioner $S \in \RR^{n \times s}$ satisfying \begin{equation} \Pr_{i \in [n-r]}\left[M_i^T \Thetat^{-1} M_i < \alpha \norm{S^T M_i}_1^2\right] \leq \frac{b}{n-r}\label{eq:alpha-bad-eqn-bound-2}\end{equation}
and \begin{equation} 128n^{7/2}k\eta \norm{M}_F \sqrt{\frac{\epsilon\beta^{(1)}_{\Thetat^{-1},S,k,m}}{\alpha\lambda}} \leq 1,\label{eq:small-epsilon-bound}\end{equation}
we have the following: with probability at least $$1 - 2/n - kb'/n - \exp(-k(r-2m)/n)$$ over true signals drawn from $\mathcal{D}_k$, it holds that $S$-preconditioned Lasso fails with probability at least $1 - (4m)/(3r) - \exp(-\Omega(m))$ over independent samples $X_1,\dots,X_m \sim N(0,\Thetat^{-1})$.
\end{lemma}

\begin{proof}
Let $S : n \times s$ be an arbitrary preconditioning matrix satisfying (\ref{eq:alpha-bad-eqn-bound-2}) and (\ref{eq:small-epsilon-bound}). By (\ref{eq:alpha-bad-eqn-bound-2}), we can apply Lemma~\ref{lem:C-exists} to define sets $C,D$ as functions of $S$, with the following properties:
\begin{itemize}
    \item For every $i \in D$, the column $v = (S^T)_i \in \RR^n$ satisfies $\norm{v}_0 \geq \tau$,
    \item The submatrix $S_{C^cD^c}$ satisfies
    $$\sum_{j \in D^c} \norm{S_{C^c j}}_2 \leq \frac{n^{3/2}\eta\norm{M}_F}{\sqrt{\lambda\alpha}}$$
    \item $|C| \leq b'$.
\end{itemize}

Draw $w^* \sim \mathcal{D}_k$. Let $K = \supp(w^*)$ and let $U = \supp(S^T w^*)$. Recall that for every $i \in D$, the $i$th row of $S^T$ (the vector $(S^T)_i \in \mathbb{R}^n$) has at least $\tau$ nonzero entries. Therefore, for a particular choice of $i \in D$, the probability that $(S^T w^*)_i = 0$ is at most the probability that a uniformly random set of $k$ elements from $[n]$ misses all $\tau$ elements of the support of $(S^T)_i$, which is at most 
\[ (1 - \tau/n)(1 - \tau/(n - 1)) \cdots (1 - \tau/(n - k)) \le e^{-k\tau/n} \le 1/(sn). \]
Hence by the union bound over all $i \in D$, recalling that $D \subset [s]$, we have that $D \subseteq U$ with probability at least $1 - 1/n$.

Similarly, because $|C| \le b'$, with probability at least $1 - kb'/n$ it holds that $K \subseteq C^c$. By Lemma~\ref{lemma:large-row-fraction}, with probability at least $1 - (1 - (r-2m)/n)^k \geq 1 - e^{-k(r-2m)/n}$ it holds that $\norm{S_{j^*}}_1 \geq 1/\sqrt{\beta\epsilon n}$ for some $j^* \in K$. Now conditioned on $K$, observe that $(w^*)_K$ has independent $N(0,1)$ entries. So by Lemma~\ref{lemma:random-sum-to-max}, $$\Pr[\norm{S^T w^*}_1 \geq \norm{S_{j^*}}_1/n] \geq 1-1/n.$$
Thus, it follows that $\norm{S^T w^*}_1 \geq 1/(n\sqrt{\beta\epsilon n})$ with probability at least $1-1/n-\exp(-k(r-2m)/n)$ over $w^*$. Moreover $\norm{w^*}_2 \leq 2\sqrt{k}$ with probability at least $1 - \exp(-\Omega(k))$. Assume for the rest of the proof that all of the above events (on $w^*$) occur: $D \subseteq U$, $K \subseteq C^c$, $\norm{S^Tw^*}_1 \geq 1/(n\sqrt{\beta\epsilon n})$, and $\norm{w^*}_2 \leq 2\sqrt{k}$. We have shown that these together occur with probability at least $1 - 2/n - kb'/n - \exp(-k(r-2m)/n)$, and in the rest of the proof we show that under these events, $S$-preconditioned Lasso fails with probability at least $1-(4m)/(3r) - \exp(-\Omega(m))$ over samples $X_1,\dots,X_m$.

Let $z = \sign(S^T w^*)$. Since $\supp(w^*) = K \subseteq C^c$, we have by Cauchy-Schwarz that
\begin{equation}\norm{(Sz)_{C^c}}_2 \geq \frac{\langle w^*, Sz\rangle}{\norm{w^*}_2} = \frac{\langle S^T w^*, z\rangle}{\norm{w^*}_2} = \frac{\norm{S^T w^*}_1}{\norm{w^*}_2} \geq \frac{1}{2n\sqrt{\beta\epsilon nk}}\label{eq:sz-norm-lb}\end{equation}
where the last inequality uses the above bounds on $\norm{S^T w^*}_1$ and $\norm{w^*}_2$. Define $$d = \argmin_{x \in \ker X_{[m],C^c}} \norm{x - (Sz)_{C^c}}_2,$$ implicitly zero-extending $d$ from $C^c$ to $[n]$. Then $Xd = X_{[m],C^c}d_{C^c} = 0$ by construction, and moreover
\begin{align}
\norm{(S^T d)_{U^c}}_1 
&\leq \norm{(S^T d)_{D^c}}_1 \tag{$D \subseteq U$} \\
&= \sum_{j \in D^c} |\langle (S^T)_{jC^c}, d_{C^c}\rangle| \tag{$\supp(d) \subseteq C^c$} \\
&\leq \norm{d}_2 \sum_{j \in D^c} \norm{S_{C^c j}}_2 \tag{Cauchy-Schwarz} \\
&\leq \norm{d}_2 \cdot \frac{n^{3/2}\eta\norm{M}_F}{\sqrt{\lambda\alpha}}. \nonumber
\end{align}
On the other hand, $$\norm{d}_2^2 = \langle S^T d, z \rangle = \norm{\Proj_{\ker X_{[m],C^c}} (Sz)_{C^c}}_2^2.$$ 
We now apply Lemma~\ref{lemma:projection-restricted-coordinates} to vector $Sz$ and set $C^c$, using that $\ker(M)$ is $(1/(12\sqrt{n}),1/(2\sqrt{n}),b')$-quantitatively dense in conjunction with the bounds $\epsilon \leq c\lambda/n^2$, $m \geq c_m\log n$, $r > 2m$, and $|C| \leq b'$. By this lemma and by (\ref{eq:sz-norm-lb}), with probability at least $1 - (4m)/(3r) - \exp(-\Omega(m))$, we can lower bound the norm of the projection to get $$\norm{d}_2 \geq \frac{1}{64\sqrt{n}}\norm{(Sz)_{C^c}}_2 \geq \frac{1}{128n^2\sqrt{\beta\epsilon k}}.$$
As a result, so long as $n^{3/2}\eta\norm{M}_F/\sqrt{\lambda\alpha} < 1/(128 n^2 \sqrt{\beta\epsilon k})$, which holds by (\ref{eq:small-epsilon-bound}), we have that $\norm{(S^T d)_{U^c}}_1 < \langle S^T d, z \rangle$ so the preconditioned Lasso fails on $w^*$.
\end{proof}

\begin{theorem}\label{theorem:all-preconditioner-failure}
There are constants $c,c_m>0$ so that the following holds. Let $M \in \RR^{n-r\times n}$, and define $\Theta = M^T M$. Let $k,k',m,\tau,b,b',\eta,t > 0$. Let $\lambda$ be the smallest non-zero eigenvalue of $\Theta$. Suppose that $M$ satisfies $(b,b',\eta,\tau)$-erasure-robustness, and that $\ker(M)$ is $(1/(12\sqrt{n}),1/(2\sqrt{n}),b')$-quantitatively-dense.
Let $\epsilon > 0$ and define $\Thetat = \Theta + \epsilon I$. Suppose that $544n^{15/2}k\eta\norm{M}_F \sqrt{\epsilon/\lambda} \leq 1$, that $r > 2m$, that $m \geq c_m\log n$, and that $\epsilon < c\lambda/n^2$. And suppose that the rows of $M$ are $k'$-sparse. Then there is a distribution $\mathcal{D}$ over $\max(k't,k)$-sparse signals such that for any preconditioner $S \in \RR^{n\times s}$ with $(n/\tau)\log(sn) < k$, with probability at least $$\frac{1}{2}\min\left(1 - e^{-bt/(n-r)}-\frac{1}{n}, 1 - \frac{2}{n} - \frac{kb'}{n} - e^{-k(r-2m)/n}\right)$$ over $w^* \sim \mathcal{D}$, the $S$-preconditioned Lasso fails to recover $w^*$ with probability at least $1 - (4m)/(3r) - \exp(-\Omega(m))$ over the samples $X_1,\dots,X_m$.
\end{theorem}

\begin{proof}
Let $\mathcal{D}_M$ be the signal distribution where we draw independent and uniformly random indices $R_1,\dots,R_t \in [n-r]$ as well as independent $Z_1,\dots,Z_t \sim \text{Unif}([-1,1])$, and set the signal to be $$w^* = \sum_{i=1}^t \frac{Z_i M_{R_i}}{\sqrt{M_{R_i}^T \Thetat^{-1} M_{R_i}}}.$$ We define $\mathcal{D}$ to be the mixture which with probability $1/2$ draws a signal from $\mathcal{D}_M$, and with probability $1/2$ draws from the distribution $\mathcal{D}_k$ defined in Lemma~\ref{lemma:compatible-preconditioner-failure} (note that in the former case, the signal is $k't$-sparse, and in the latter case, it is $k$-sparse). 

Pick any $S \in \RR^{n\times s}$ with $(n/\tau)\log(sn) < k$. Define $\alpha = \beta^{(1)}_{\Thetat^{-1},S,k,m}/(18n^4)$. We distinguish two cases. 

\paragraph{Case I: incompatible preconditioner.} On the one hand, suppose that $$\Pr_{i \in [n-r]}\left[M_i\Thetat^{-1}M_i < \alpha \norm{S^T M_i}_1^2\right] \geq \frac{b}{n-r}.$$ Then with probability at least $$\frac{1}{2}\left(1 - (1 - b/(n-r))^t\right) \geq \frac{1}{2}\left(1 - e^{-tb/(n-r)}\right)$$ over the row indices $R_1,\dots,R_t$, there is some $R_i$ with $$\norm{S^T M_{R_i}}_1^2 > \frac{1}{\alpha} M_{R_i}^T\Thetat^{-1} M_{R_i}.$$ 
In this event, we have $$\max_{i \in [t]} \frac{\norm{S^T M_{R_i}}_1}{\sqrt{M_{R_i}^T \Thetat^{-1} M_{R_i}}} > \frac{1}{\sqrt{\alpha}},$$
so by Lemma~\ref{lemma:random-sum-to-max}, it holds with probability at least $1-1/n$ over $Z_1,\dots,Z_t$ that $\norm{S^T w^*}_1 \geq \frac{1}{n\sqrt{\alpha}}.$ But by the triangle inequality, we have $\sqrt{(w^*)^T \Thetat^{-1} w^*} \leq t \leq n$. Thus, $$\Pr_{w^* \sim \mathcal{D}}\left[\norm{S^T w^*}_1^2 > \frac{1}{n^4\alpha}(w^*)^T \Thetat^{-1} w^*\right] \geq \frac{1}{2}\left(1 - e^{-tb/(n-r)}\right)\left(1 - \frac{1}{n}\right).$$

Moreover, for such $w^*$, by choice of $\alpha$ and by Theorem~\ref{theorem:ill-conditioned-lasso-failure-l1}, the $S$-preconditioned Lasso recovers $w^*$ with probability at most $\exp(-\Omega(m))$, from $m$ independent samples $X_1,\dots,X_m \sim N(0,\Thetat^{-1})$.

\paragraph{Case II: compatible preconditioner.} On the other hand, suppose that $$\Pr_{i \in [n-r]}\left[M_i\Thetat^{-1}M_i < \alpha \norm{S^T M_i}_1^2\right] \leq \frac{b}{n-r}.$$
This is precisely the condition (\ref{eq:alpha-bad-eqn-bound-2}) in Lemma~\ref{lemma:compatible-preconditioner-failure}. Moreover, (\ref{eq:small-epsilon-bound}) is satisfied because we defined $\alpha$ so that $\beta^{(1)}_{\Thetat^{-1},S,k,m}/\alpha = 18n^4$, and because of the bound assumed in the theorem statement. Thus, we can invoke Lemma~\ref{lemma:compatible-preconditioner-failure}, which states that for $w^* \sim \mathcal{D}_k$, with probability at least $1 - 2/n - kb'/n - \exp(-k(r-2m)/n)$ over the signal, the $S$-preconditioned Lasso fails with the desired probability. Thus, the same holds for $w^* \sim \mathcal{D}$ up to a factor of $2$.
\end{proof}


Finally, we instantiate the above theorem with the matrix constructed in Theorem~\ref{theorem:existence}.

\begin{theorem}\label{theorem:main}
Let $n \in \NN$ be sufficiently large. There is a matrix $\Sigma \in \RR^{n \times n}$ with condition number $\poly(n)$ such that the following holds. Let $k \in \NN$ with $k \geq \log^8 n$. There is a distribution $\mathcal{D}$ over $k\log^4 n$-sparse signals such that for any positive integer $m \leq n/5$ and any preconditioner $S \in \RR^{n \times s}$ with $s < \exp(ck/\log^6 n - \log n)$, with probability at least $\frac{1}{2}-O(1/\log n)$ over $w^* \sim \mathcal{D}$, the $S$-preconditioned Lasso recovers $w^*$ with probability at most $8m/3n + \exp(-\Omega(m))$ from $m$ independent samples $X_1,\dots,X_m \sim N(0,\Sigma)$ with noiseless responses $Y_i = \langle X_i, w^*\rangle$.
\end{theorem}

\begin{proof}
Let $\Theta$ be the matrix guaranteed by Theorem~\ref{theorem:existence}. We check the conditions of Theorem~\ref{theorem:all-preconditioner-failure}. Clearly, $r := \dim \ker(\Theta) \geq n/2$. By claim (5), we have $\lambda = \Omega(n^{-5/2})$. By claims (3) and (2), we can take $b = n/(k\log n)$, $b' = 2n/(k\log n)$, $\eta = n^6\log^8 n$, and $\tau = \Omega(n/\log^6 n)$. By claim (4), we have $\norm{M}_F \leq O(\sqrt{n\log n})$. Thus, we can take $\epsilon = \Omega(n^{-22})$. By claim (1), the rows of $M$ are $k'$-sparse with $k' = O(\log^2 n)$. Let $t = k\log^2 n$.

So applying Theorem~\ref{theorem:all-preconditioner-failure}, we get that there is a distribution $\mathcal{D}$ over $k\log^4 n$-sparse signals such that for any preconditioner $S \in \RR^{n \times s}$ with $s < \exp(ck/\log^6 n - \log n)$, with probability at least $$\frac{1}{2} \min\left(1 - \frac{2}{n}, 1 - \frac{2}{n} - \frac{2}{\log n} - e^{-k(n/2 - 2m)/n}\right)$$ over $w^* \sim \mathcal{D}$, the $S$-preconditioned Lasso recovers $w^*$ with probability at most $8m/(3n) + \exp(-\Omega(m))$ from $m$ samples. So long as $m \leq n/5$, both terms in the minimum are $1 - O(1/\log n)$, yielding the claimed result.
\end{proof}

\bibliographystyle{amsalpha}
\bibliography{bib}
\newpage


\end{document}